\pdfoutput=1
\newif\ifdraft \draftfalse

\drafttrue 

\documentclass[11pt]{article}
\usepackage[utf8]{inputenc}
\usepackage{fullpage}
\usepackage{subcaption}
\usepackage{graphicx} 
\graphicspath{{figures/}}
\usepackage{amssymb}

\usepackage{parskip} 
\usepackage{url} 
\usepackage{natbib} 
\bibpunct{(}{)}{;}{a}{,}{,} 

\usepackage{amsthm, amsmath,amssymb,nicefrac}
\usepackage{algorithm}
\usepackage[noend]{algorithmic}\usepackage{color}
\usepackage{bbm}
\newtheorem{theorem}{Theorem}
\newtheorem{definition}[theorem]{Definition}
\newtheorem{fact}[theorem]{Fact}
\newtheorem{corollary}[theorem]{Corollary}
\newtheorem{proposition}[theorem]{Proposition}
\newtheorem{lemma}[theorem]{Lemma}
\newtheorem{claim}[theorem]{Claim}
\newcommand{\calM}{\ensuremath{\mathcal{M}}}
\newcommand{\calN}{\ensuremath{\mathcal{N}}}
\newcommand{\calP}{\ensuremath{\mathcal{P}}}
\newcommand{\calQ}{\ensuremath{\mathcal{Q}}}
\newcommand{\calD}{\ensuremath{\mathcal{D}}}
\newcommand{\calX}{\ensuremath{\mathcal{X}}}
\newcommand{\LR}{\texttt{LR}}
\renewcommand{\Pr}{\mathbb{P}}
\newcommand{\E}{\ensuremath{ \mathbb{E} }}

\usepackage{macros}
\newcommand{\cut}[1]{}

\title{Locally Private Mean Estimation: $Z$-test and Tight Confidence Intervals}
\author{Marco Gaboardi\\Dept. of Computer Science\\University of Buffalo\\{\tt gaboardi@buffalo.edu} \and Ryan Rogers\ \\{\tt rrogers386@gmail.com}\and Or Sheffet\\Dept. of Computing Science\\Univresity of Alberta\\{\tt osheffet@ualberta.ca}}
\begin{document}

\maketitle

\begin{abstract}
This work provides tight upper- and lower-bounds for the problem of mean estimation under $\epsilon$-differential privacy in the local model, when the input is composed of $n$ i.i.d. drawn samples from a normal distribution with variance $\sigma$. Our algorithms result in a $(1-\beta)$-confidence interval for the underlying distribution's mean $\mu$ of length  
$\tilde O\left( \nicefrac{\sigma \sqrt{\log(\nicefrac 1 \beta)}}{\epsilon\sqrt n} \right)$. 
In addition, our algorithms leverage binary search using local differential privacy for quantile estimation, a result which may be of separate interest.
Moreover, we prove a matching lower-bound (up to poly-log factors), showing that any one-shot (each individual is presented with a single query) local differentially private algorithm must return an interval of length $\Omega\left( \nicefrac{\sigma\sqrt{\log(1/\beta)}}{\epsilon\sqrt{n}}\right)$.
\end{abstract}

\section{Introduction}
\label{sec:intro}

In the last decade, differential privacy~\citep{DMNS06} has become the de-facto gold standard of privacy-preserving data analysis. Moreover, in recent years, there has been a growing interest in devising differentially private techniques for statistical inference (see Related Work below). However, by and large, these works have focused on the centralized model, where the dataset in its entirety is given to a trusted curator who has direct access to the data. This is in contrast to the trust-free \emph{local} model~\citep{Warner65,KLNRS08}, in which each individual perturbs her own data and broadcasts the noisy (and privacy preserving) outcome.  The local model is growing in popularity in recent years with practical, large scale deployments (see \cite{EPK14,appleDP2016}). 
Yet only a handful of works~\citep{DJW13FOCS, DJW13NIPS, GR18, Sheffet18} examine differentially private statistical inference techniques in the local-model.

This work focuses on the task of \emph{mean estimation} in the local-model. The problem is composed of $n$ i.i.d samples drawn from a Gaussian $X_1,...,X_n \stackrel{\rm i.i.d.}{\sim} \calN(\mu,\sigma^2)$ such that $\mu\in [-R,R]$ for some known bound $R$, and $\sigma$ is either provided as an input (known variance case) or left unspecified (unknown variance case). We point out that the privacy analysis in our algorithms hold even if the assumption of normal data is not satisfied, whereas our utility analysis relies on this assumption.  The goal of our algorithms is to provide an estimation of $\mu$, which may be represented in multiple forms. The classical approach in statistical inference is to represent the likelihood that each point on the real line is $\mu$ with a probability distribution~--- where in the case of known variance  ($Z$-test) the output is a Gaussian distribution, and in the case of unknown variance ($T$-test) the output is a $t$-distribution. This distribution allows an analyst to estimate a \emph{confidence interval} $I$ based on the random sample of data s.t. $\prob{\mu\in I}\geq 1-\beta$, where non-privately it holds that $|I|=O(\nicefrac \sigma {\sqrt n})$ (assuming $\beta$ is a constant). Based on confidence intervals, one is able to reject (or fail-to-reject) certain hypotheses about $\mu$, such as the hypothesis that $\mu=0$ or that the means of two (or more) separate collections of samples ($X_1,...,X_n$ and $ Y_1,...,Y_m$) are identical.

\paragraph{Our Contribution.} The goal of this work is to provide upper- and lower-bounds for the problem of mean-estimation under $(\epsilon,\delta)$-local differentially private (LDP) algorithms
assuming the data is drawn from an unknown Gaussian. For our upper bounds in the case of known variance, we design a $\epsilon$-LDP algorithm,
which yields a confidence interval of length $O(\sigma \cdot \nicefrac {\sqrt{\log(n)}}{\epsilon\sqrt n})$ provided that $n=\Omega(\frac{\log(\nicefrac R \sigma)}{\epsilon^2})$.  In the case of unknown variance we give an algorithm that returns a valid confidence interval of similar length assuming we have a lower bound on the value of the unknown $\sigma$. For our lower-bounds, we prove that any $\epsilon$-LDP algorithm must return an interval whose length is $\Omega(\nicefrac {\sigma}{\epsilon\sqrt n})$, proving the optimality of our technique up to a $\sqrt{\log(n)}$-factor. {In the known variance case, our algorithm results in a private $Z$-test, which we also assess empirically.} 

\subsection{Our Techniques: Overview}
\label{sec:techniques_overview}

\paragraph{Basic Tools.}
In our algorithms, we use two basic LDP canonical algorithms of \emph{Randomized Response}~\citep{Warner65, KLNRS08} and \emph{Bit Flipping} (in its various versions)~\citep{EPK14,BS15,BNST17}. The mechanisms are known, and, for completeness, in Section~\ref{sec:preliminaries} we provide utility bounds for these building blocks under randomly drawn input. 
%
%


\paragraph{The Known Variance Case.}
In the known variance case, our approach is a direct LDP implementation of the ideas behind the algorithm of~\citet{KV18} who provide a private confidence interval in the centralized model. We equipartition the interval where $\mu$ is assumed to be between $[-R,R]$ into $d =\left\lceil \frac{2R}{\sigma} \right\rceil$ sub-intervals of length $\sigma$, and use the above-mentioned Bit Flipping mechanism to find the most likely interval. 
The most common interval must be within distance $2\sigma$ from the mean (with high probability) of the underlying Gaussian distribution. This allows us to narrow in on an interval $I$ of length $O\left(\sigma\sqrt{\log(n/\beta)}\right)$ which should hold $n$ new points from the same distribution with probability at least $1-\beta$. 


Once we have found this interval, we merely project each datapoint onto $I$ 
\ifdefined \toggleLaplace
add Laplace noise with the appropriate scale $O(\sigma\sqrt{\log(n/\beta)} / \epsilon)$ to ensure $\epsilon$-LDP, and average the outcomes.
We know that with probability $2\beta$ that the average of all the Laplace random variables will be at most $O\left( \frac{\sigma\sqrt{\log(n/\beta)\log(1/\beta)}}{\sqrt{n}\epsilon} \right)$ and the $n$ Gaussian terms will all be within $I$.
\else
and add Gaussian noise of $\calN\left( 0, \frac {2|I|^2 \log(2/\delta)}{\epsilon^2} \right)$ 
to the projection, and then average the outcomes. 
This implies we have $n$ i.i.d sample points for a Gaussian of mean $\mu$ and variance $O\left( \frac{\sigma^2 \log(n/\beta)\log(1/\delta)}{\epsilon^2} \right)$.\footnote{Actually, this is an approximation of the distribution, since we clip the original Gaussian. However, since the probability mass we remove is $<\frac \beta n$, the TV-dist to this distribution is $<1/n$.} Thus, $\tilde \mu$, the average of these $n$ noisy datapoints, is also sampled from a Gaussian, whose variance is $\tilde \sigma^2 = O\left( \frac{\sigma^2 \log(n/\beta)\log(1/\delta)}{\epsilon^2n} \right)$. We can thus represent the likelihood that each point on $\mathbb{R}$ is the mean by using a Gaussian $\calN(\tilde \mu, \tilde \sigma^2)$, which is our analog to the $Z$-test. Moreover, the interval of length $2\tilde\sigma\sqrt{\log(4/\beta)}$ centered at $\tilde{\mu}$ is a  $(1-\beta)$-confidence interval. Details appear in Section~\ref{sec:known_variance}, where in Section~\ref{subsec:Ztest} we also present some empirical assessment of our $Z$-test.
\fi

\paragraph{The Unknown Bounded Variance Case.}
We now consider the case of unknown variance, where instead of knowing $\sigma$ we are provided bounds on the smallest and largest (resp.) values of the variance: $\sigma_{\min}, \sigma_{\max}$. First, we illustrate our algorithm in the case where we know $\sigma_{\max} \leq 2R$. This is of course the more natural case, as we think of $R$ as large and $\sigma$ as reasonable. Later, we discuss how to deal with the case of general unknown variance.

In this case, the approach of \citet{KV18} is to estimate the variance using the pairwise differences of the datapoints. That is due to the property of Gaussians where the difference between two i.i.d samples is $\calN(0,2\sigma^2)$. This however is an approach that only works in the centralized model, where one is able to observe two datapoints without noise. In the local model, we are forced to use a different approach. 

The approach we follow is to do binary search for different quantiles of the Gaussian, an approach which has appeared before in certain testers, and in particular in the work of~\citet{Feldman17}. 
Given a quantile $p\in (0,1)$, a continuous and smooth distribution $\calP$, our goal is to find the threshold point $t$ such that $p - \lambda\leq \Prob{X\sim\calP}{X<t} \leq p + \lambda$ for a given tolerance parameter $\lambda>0$. In each iteration $j$, we hold an interval $I^{(j)}$ which is guaranteed to hold $t$, and we use the middle point of this interval as our current guess. Denoting $t^{(j)}$ as the current interval's mid-point, we use enough of the dataset to estimate $\Prob{X\sim\calP}{X<t^{(j)}}$ up to error $\lambda$, and then either halt (if the estimated probability is approximately $p$) or recurse on either the left- or right-half of the interval. Since our initial interval is $[-(R+\sigma_{\max}),R+\sigma_{\max}]$ (of length $<6R$) and we must halt when we reach an interval of length $\Omega(\sigma_{\min})$ (we treat $\lambda$ as a constant), then the number of iterations overall is $T = O(\log(R/\sigma_{\min}))$.

And so, we first run binary search till we find a point $t_1$ for which we estimate that $\Prob{X\sim \calN(\mu,\sigma^2)}{X<t_1} \approx 50\%$. We then find a point $t_2$ for which we estimate that $\Prob{X\sim \calN(\mu,\sigma^2)}{X<t_2} \approx 81.4\%$. Due to the properties of a Gaussian, $t_1\approx \mu$ and $t_2-t_1\approx \sigma$. Of course, we do not have access to the actual quantiles, but rather just an estimation of them, but we are still able to show that w.p. $\geq 1-\beta$ it holds that $0.5\sigma <t_2-t_1 < 2\sigma$. (These bounds explain why taking $\lambda$ as a constant, say $\lambda=0.05$, suffice for our needs.) We can thus run the algorithm for the known variance case with this estimation of the variance on the remainder of the dataset. The full details of our algorithm appear in Section~\ref{sec:unknown_variance_algorithm}.

\paragraph{The General Unknown Variance Case.}
In the general case, where $\sigma_{\max}$ isn't known, we begin by testing to see if the variance is $\geq R$ or $<2R$ by estimating the probability that a new datapoint falls inside the interval $[-2R,2R]$. If this probability is large then we have that $\sigma<2R$ and we can use the previous algorithm for unknown bounded variance; whereas if this probability is small, it must be that $\sigma>R$, and we run a very different algorithm. Instead of binary search, we merely estimate $q_1 \stackrel{\rm def}= \Prob{X\sim\Normal{\mu}{\sigma^2}}{X<-R}$ using the first half of the points, and then estimate $q_2 \stackrel{\rm def}= \Prob{X\sim\Normal{\mu}{\sigma^2}}{X<R}$ using the remaining half of the points. 
Denoting $t_1$ and $t_2$ as the points on the real line for which the {\sf CDF} of a standard normal $\Normal{0}{1}$ equals $q_1$ and $q_2$ respectively, we can now interpolate a Gaussian curve that matches $-R$ to $t_1$ and $R$ to $t_2$ and infer its mean and variance accordingly. 
The key point is that both $-R$ and $R$ are within distance $<2\sigma$ of the true mean $\mu$; so by known properties of the Gaussian distribution, estimating $q_1$ and $q_2$ up to an error of $O(\nicefrac1 {\epsilon\sqrt{n}})$ implies a similar error guarantee in estimating~$\mu$. This approach is discussed in Section~\ref{sec:large_unknown_variance}.

\paragraph{Lower Bounds}
Lastly, we give bounds on any $\epsilon$-LDP algorithm that approximates the mean of a Gaussian distribution. Formally, we say an algorithm \emph{$(\beta,\tau)$-solves the mean-estimation problem} if its input is a sample of $n$ points drawn i.i.d from a Gaussian distribution $\Normal{\mu}{\sigma^2}$ with $\mu\in[-R,R]$ for some given parameter $R$, and its output is an interval $I$ such that $\mu\in I$ w.p. $\geq 1-\beta$ and furthermore $\ex{|I|} \leq \tau$. Note that the probability is taken over \emph{both} the sample draws and the coin-tosses of the algorithm.
We prove that any one-shot, where each datapoint is queried only once, $\epsilon$-locally differentially private algorithm $\calM$  that $(\beta,\tau)$-solves that mean estimation problem must have that $\tau \in \Omega\left( \tfrac{\sigma\sqrt{\log(1/\beta)}}{\epsilon\sqrt{n}}\right)$ and also hold that $n\in\Omega\left( \tfrac 1 {\epsilon^2} \log( \tfrac R{\beta\cdot\tau} )\right)$. 

In addition, we also provide lower bounds for any one-shot $\epsilon$-LDP algorithm that approximates the quantile of a given distribution $\calP$ using i.i.d samples from $\calP$. Our bounds show that dependency on certain parameters ($R, \sigma_{\min}$) is necessary. In particular, if $R$ (or $\sigma_{\min}$) is left unspecified (namely $R=\infty$ and $\sigma_{\min}=0$), then no LDP algorithm can $(\beta,\tau)$-solve the mean-estimation problem.

Note that our upper-bounds are given by $(\epsilon,\delta)$-LDP algorithms, yet our lower bounds deal only with $\epsilon$-LDP algorithms. However, a recent result of~\cite{BNS18} shows that in the local model (as opposed to the centralized model) any $(\epsilon,\delta)$-LDP is equivalent to a $\epsilon$-LDP algorithm. Further details appear in the Preliminaries.

\subsection{Related Work}
\label{subsec:related_work}

Several works have studied the intersection of differential privacy and statistics~\citep{DworkL09, Smith11, ChaudhuriH12, DJW13FOCS, DJW13NIPS, DSZ15} mostly focusing on robust statistics; but only a handful of works study rigorously the significance and power of hypotheses testing under differential privacy~\citep{VuS09, USF13, WLK15, GLRV16, KR17, CDK17, Sheffet17, KV18}. \citet{VuS09} looked at the sample size for privately testing the bias of a coin. \citet{JS13}, \citet{USF13} and~\citet{YFSU14} focused on the Pearson $\chi^2$-test, showing that the noise added by differential privacy vanishes asymptotically as the number of datapoints goes to infinity, and propose a private $\chi^2$-based test which they study empirically. \citet{WLK15}, \citet{GLRV16}, and~\citet{KR17} then revised the asymptotic statistical tests themselves to incorporate the additional noise due to privacy as well as the randomness in the data sample.  \citet{AcharyaSuZh18} and \citet{AliakbarpourDR18} give sample complexity upper and lower bounds in identity and equivalence testing. 
\citet{CDK17} gives a private identity tester based on noisy $\chi^2$-test over large bins. \citet{Sheffet17} studies private Ordinary Least Squares using the JL transform.  All of these works however deal with the centralized-model of differential privacy.

Few additional works are highly related to this work. 
\citet{KV18} give matching upper- and lower-bounds on the confidence intervals for the mean of a population, also in the centralized model.  See also work from \citet{KamathLiSiUl18}, which gives a multivariate extension of estimating Gaussians in the central privacy model.
\citet{DJW13FOCS, DJW13NIPS} give matching upper- and lower-bound on robust estimators in the local model, and in particular discuss mean estimation. However, their bounds are related to minimax bounds rather than mean estimation or $Z$-tests. \citet{GR18} and~\citet{Sheffet18} study the asymptotic power and the sample complexity (respectively) of a variety of $\chi^2$-squared based hypothesis tests in the local model.  Recent work from \citet{AcharyaCaFrTy18} improves on the sample complexity bounds in \citet{Sheffet18} and consider the setting where access to public randomness is allowed.  There is also work from \citet{KairouzBoRa16} on estimating distributions for categorical data in the local privacy model.  Finally, we mention the related work of~\citet{Feldman17} who also discusses mean estimation using a version of a statistical query oracle which is thus related to LDP. Similar to our approach, \citet{Feldman17} also uses the folklore approach of binary search in the case the input variance is significantly smaller than the given bounding interval.

\section{Preliminaries}
\label{sec:preliminaries}

We will write the dataset $\bX \stackrel{i.i.d.}{\sim} \Normal{\mu}{\sigma^2}$ where $\bX = (X_1, \cdots, X_n)$.  Our goals is to develop confidence intervals for the mean $\mu$ subject to local differential privacy in two settings: (1) known variance, (2) unknown variance.  We assume that the mean $\mu$ is in some finite interval $\mu \in [-R,R]$ and similarly for the standard deviation $\sigma \in [\sigma_{\min}, \sigma_{\max}]$, if it is not known a priori.   We first present the definition of differential privacy in the \emph{curator} model, where the algorithm takes a single element from universe $\cX$ as input.
\begin{definition}[\cite{DMNS06,DKMMN06}]
An algorithm $\cM: \cX \to \cY$ is $(\epsilon,\delta)$-differentially private (DP) if for all $x, x' \in \cX$ and for all outcomes $S \subseteq \cY$, we have
$$
\prob{\cM(x) \in S} \leq e^\epsilon\prob{\cM(x') \in S} + \delta.
$$
\end{definition}
We then define \emph{local} differential privacy, formalized by \cite{KLNRS08}, which does not require individuals to release their raw data to some curator, but rather each data entry is perturbed to prevent the true entry from being stored. 

\begin{definition}[LR Oracle]
Given a dataset $\bbx$, a \emph{local randomizer oracle} $\LR_{\bbx}(\cdot, \cdot)$ takes as input an index $i \in[n]$ and an $(\epsilon,\delta)$-DP algorithm $R$, and outputs $y \in \cY$ chosen according to the distribution of $R(x_i)$, i.e. $\LR_{\bbx}(i,R) = R(x_i)$.  
\end{definition}

\begin{definition}[\citet{KLNRS08}]
An algorithm $\alg: \cX^N \to \cY$ is $(\epsilon,\delta)$-\emph{local differentially private} (LDP) if it accesses the input database $\bbx \in \cX^n$ via the LR oracle $\LR_{\bbx}$ with the following restriction: if $\LR(i,R_j)$ for $j \in [k]$ are the $\alg$'s invocations of $\LR_{\bbx}$ on index $i$, then each $R_j$ for $j \in [k]$ is  $(\epsilon_j,\delta_j)$- DP and $\sum_{j=1}^k \epsilon_j \leq \epsilon$, $\sum_{j=1}^k \delta_j \leq \delta$.
\end{definition}

In this work we present and prove bounds regarding \emph{one-shot} mechanisms, where an algorithm is allowed to only query a user once and then she is never queried again.  

\begin{definition}
	\label{def:one_shot_local_DP}
	We say a randomized mechanism $\calM$~is a \emph{one-shot local differentially private} if for any dataset input $D$, $\calM$ interacts with datum $x_i$ by first choosing a single differentially private mechanism $\calM_i$, applying $\calM_i(x_i)$ and then only post-processes the resulting output without any further interaction with $x_i$. In other words, $\calM$ has only one-round of interaction with any datapoint. As a result $\calM(D)$ is merely post-processing of the length $n$ vector of outputs $( \calM_1(x_1), \calM_2(x_2),..., \calM_n(x_n) )$.
\end{definition}
Note that the definition of a one-shot mechanism does not rule out choosing the separate mechanisms \emph{adaptively}~--- it is quite possible that $\calM_i$ depends on previous outcomes $\calM_{j}(x_j)$ for $j < i$. The definition only rules out the possibility of $\calM$ revisiting the datum of an individual based on prior responses from this datum.

We now present a result from \cite{BNS18}, which shows that \emph{approximate} differential privacy, i.e. $(\epsilon,\delta)$-DP where $\delta>0$, cannot provide more accurate answers than \emph{pure}-differential privacy, i.e. $\delta =0$, in the local setting.  This is another significant difference between the local and central model due to the fact that \emph{approximate}-DP answers can be significantly more accurate than \emph{pure}-DP answers in the central model.  
\begin{theorem}[\citet{BNS18}]
Fix parameter $\eta >0$.  Let $\calM$ be $(\epsilon,\delta)$-LDP with $\epsilon \leq 1/4$ and $\delta \leq \tfrac{\epsilon\eta}{48 n \log(2n/\eta)} $.  Then there exists an algorithm $\calM'$ that is $10\epsilon$-LDP and $\calM'(\bbx)$ has total variation distance of at most $\eta$ from $\calM(\bbx)$ for any input $\bbx$.  
\label{thm:approx2pure}
\end{theorem}
This result will prove to be useful in showing that our local private confidence interval widths are tight up to polylogarithmic terms.  Note that this result was extended to other values of $\epsilon$ by \citet{CSUZZ18}.  

We next define our utility goal, which is to find confidence intervals that contain the mean parameter $\mu$ with high probability, where the probability is over the sample and the randomness of the LDP algorithm.

\begin{definition}[Confidence Interval]
An algorithm $\cM: \R^N \to 2^\R$ produces a valid $(1-\beta)$-confidence interval for the mean of the underlying Gaussian distribution $\Normal{\mu}{\sigma^2}$
if the following holds
$$
\Prob{\bX \stackrel{i.i.d.}{\sim} \Normal{\mu}{\sigma^2}, \cM(\bX)}{\mu \in \cM(\bX)} \geq 1- \beta
$$
\end{definition}
Our primary objective is to design an algorithm that is $(\epsilon,\delta)$-LDP that also produces a valid $(1-\beta)$-confidence interval.

\paragraph{Useful Bounds.} Throughout this paper, we use several concentration bounds, especially for Gaussians, where it is known that for any $\beta \in (0,\nicefrac 1 2)$ we have
\[ \Prob{X\sim\Normal{\mu}{\sigma^2}}{|X-\mu| > \sigma\sqrt{2\log(\nicefrac 2 \beta)}} \leq\beta \]

A useful tool in our analysis is the following well-known variation of  McDiarmid's inequality. The Hoeffding inequality is a direct result of it, in the case all random variables are distributed i.i.d.
\begin{fact}
	\label{fact:McDiarmid}[McDiarmid's Inequality]
	Let $X_1,...X_n$ be $n$ independent random variables. Denote $B_1,...,B_n$ and $\mu_1,...,\mu_n$ such that $\forall i, |X_i|\leq B_i$ and $\E[X_i]=\mu_i$. Then for any $t>0$ we have
	\[  \prob{ \left|\sum_i X_i - \sum_i\mu_i\right| > t    } \leq 2\exp\left( -\frac{2t^2}{\sum_i B_i^2} \right) \]
\end{fact}

\subsection{Existing Locally Private Mechanisms}
\label{subsec:basic_tools}

A basic approach to preserve differential privacy is to use additive random noise. Suppose each datum is sampled from an interval $I$ of length $\ell$. Then adding random noise taken from $\Lap(\nicefrac \ell \epsilon)$ to each datum (independently) guarantees $\epsilon$-differential privacy~\citep{DMNS06}; and adding random noise taken from $\Normal{0}{\nicefrac {2\ell^2\log(\nicefrac{2}{\delta})} {\epsilon^2}}$ to each datum (independently) guarantees $(\epsilon,\delta)$-differential privacy~\citep{DKMMN06}.

Another canonical $\epsilon$-local differentially private algorithm is the \emph{randomized response} algorithm \citep{Warner65}. In this mechanism, each datum is a bit       $\{0,1\}$ and on each datum we operate independently, applying $\RR:\{0,1\} \to \{0,1\}$ where
$$
\RR(b) = \left\{ \begin{array}{lr}
			b & \text{ w.p. } \frac{e^\epsilon}{1+e^\epsilon} \\
			1-b & \text{ else}
\end{array}
\right.
$$
It is straight-forward to see that on an input composed of $m$ many $1$s and $n-m$ many $0$s, the expected number of $1$s in the output is \[ m\cdot \frac{e^\epsilon}{1+e^\epsilon} + (n-m) \frac{1}{1+e^\epsilon} = \frac{n}{1+e^\epsilon} + m\cdot \frac{e^\epsilon-1}{1+e^\epsilon}  \]
and so the na\"ive estimator for the number of $1$s in the input is
\begin{equation}
 \hat{\theta_{\rm RR}} \stackrel{\rm def}{=}
~~~\frac{e^\epsilon +1 }{e^\epsilon-1}\sum_{i = 1}^n \RR(b_i) - \frac n {e^\epsilon-1} 
\label{eq:RR_estimator}
\end{equation}

The following claim summarizes a folklore result about input chosen i.i.d from a distribution. This will be useful in the sequel for our results. 

\begin{claim}
	\label{clm:RR_utility_random_input}
	Let $\calX$ be a domain and let $\calD$ be a distribution over this domain. Given a predicate $\phi:\calX \to\{0,1\}$, we denote $p = \Ex{X\sim \calD}{\phi(X)}$.
	Given $n$ i.i.d draws $\bX$ from $\calD$, denote by $\hat{\theta_{\rm RR}}(n,\phi)$ the randomized response estimator in~\eqref{eq:RR_estimator} applied to the $n$ bits $\phi(X_1), \phi(X_2), ..., \phi(X_n)$. 
	Fix any $\alpha,\beta \in (0,\tfrac 1 2)$. Then if $n\geq \tfrac 2 {\alpha^2} \left( \tfrac{e^\epsilon+1}{e^\epsilon-1}\right)^2\log(\tfrac 4 \beta)$ then we have that \[ \prob{ \left|\tfrac 1 n \cdot \hat{\theta_{\rm RR}}(n,\phi)- p\right|\leq \alpha  }\geq 1-\beta  \]
\end{claim}
\begin{proof}
	The proof applies both the Hoeffding and the McDiarmid inequality. Denoting $m$ as the number of $1$s in the sampled input, we argue that when $n$ is large enough we have that
	\[ \prob{ \left|p-\tfrac m n\right| > \tfrac \alpha 2  }\leq \tfrac \beta 2 \qquad{\rm and}\qquad \prob{ \left| \hat{\theta_{\rm RR}}(n,\phi)-  m \right|> \tfrac {\alpha n} 2  }\leq \tfrac \beta 2  \]
	The first of the two inequalities is an immediate consequence of the Hoeffding bound, stating that in the process of sampling the $n$ entries from the distribution, $\prob{ |p-\tfrac m n| > \tfrac \alpha 2  }\leq 2\exp(-2n\alpha^2/4) \leq \tfrac \beta 2$ since $n\geq  \tfrac 2 {\alpha^2} \log(\nicefrac 4 \beta)$. Having fixed the input to have exactly $m$ ones, it is evident that $\hat{\theta_{\rm RR}}$ is a function of the $n$-bit input $\phi(x_1),...,\phi(x_n)$, with $\E[\hat{\theta_{\rm RR}}] = m$ and where each datum can affect its value by at most $\tfrac{e^\epsilon+1}{e^\epsilon-1}$. McDiarmid's inequality thus states that 
	\begin{align*} \prob{ \left|\hat{\theta_{\rm RR}}-m\right|> \tfrac {\alpha n} 2   } &\leq 2 \exp\left( -\tfrac {2\alpha^2 n^2 / 4} {n \cdot\left( \frac{e^\epsilon+1}{e^\epsilon-1}\right)^2}\right) 
	\cr &= 2\exp\left( -n \frac{\alpha^2} {2\left( \tfrac{e^\epsilon+1}{e^\epsilon-1}\right)^2}  \right) 
	\leq \tfrac \beta 2 \end{align*} as $n\geq \tfrac 2 {\alpha^2} \left( \tfrac{e^\epsilon+1}{e^\epsilon-1}\right)^2\log(\nicefrac 4\beta) $.
\end{proof}

Another useful local differentially private algorithm is the bit flipping algorithm~\citep{EPK14, BS15}.  Let $\cX$ be a domain and let $\phi:\cX \to \{1,2,..,d\}$ be a partition of $\cX$ into $d$ types. This allows us to identify each datum $x_i$ in our dataset with a $d$-dimensional vector indicating the type $\phi(x_i)$ using a standard basis vector, or one-hot vector. The Bit Flipping mechanism now runs $d$ independent randomized response mechanism for each coordinate separately, where the privacy-loss for each coordinate is set as $\nicefrac \epsilon 2$. Therefore, per datum we output a vector $\pmb V_i\in\{0,1\}^d$, and seeing as each coordinate is slightly skewed towards $0$ or $1$, then de-biasing with the following estimator is likely to produce a good approximation of the true histogram for the input dataset:
\begin{equation}
\label{eq:bit_flippiing_estimator}
\hat{\theta_{\rm BF}} \stackrel{\rm def}{=}\sum_i \frac{e^{\epsilon/2} + 1}{e^{\epsilon/2}-1} \cdot \left(\pmb {V}_i - \frac{1}{1+e^{\epsilon/2}}\cdot \pmb{1}\right)
\end{equation} 
Again, our focus is on the performance of the bit flipping mechanism over random input. Specifically, in the sequel we will used the following property.

\begin{claim}
	\label{clm:BF_utility_random_input}
	Let $\calX$ be a domain and let $\calD$ be a distribution over this domain. Given a domain partition $\phi:\calX \to\{1,2,..,d\}$, we denote $\pmb p$ as the vector whose $j$th entry is $p_j = \Ex{X\sim \calD}{\phi(X)=j}$.
	Given $X_1,...,X_n \stackrel{i.i.d.}{\sim}\calD$, we denote the bit-flipping histogram applied to the $n$ $d$-dimensional standard-basis vectors  $\pmb e_{\phi(X_1)}, \pmb e_{\phi(X_2)},...,\pmb e_{\phi(X_n)}$.	
	Fix any $\alpha,\beta \in (0,\tfrac 1 2)$. Then if $n\geq \tfrac 2 {\alpha^2} \left( \tfrac{e^{\epsilon/2}+1}{e^{\epsilon/2}-1}\right)^2\log(\nicefrac {4d}\beta) $ then we have that \[ \prob{ \|\tfrac 1 n \cdot \hat{\theta_{\rm BF}}(n,\phi)- \pmb p\|_\infty \leq \alpha  }\geq 1-\beta  \]
\end{claim}
\begin{proof}
	The proof is similar to the proof of Claim~\ref{clm:RR_utility_random_input}, replacing the naive bounds with a union bound.
		We apply both the Hoeffding and the McDiarmid inequality. Denote the empirical histogram over the $d$ types specified by $\phi$ over the drawn $n$ inputs as $\pmb q$. We argue that when $n$ is large enough we have that $\prob{\|\pmb p- \pmb q\|_\infty > \tfrac \alpha 2  }\leq \tfrac \beta 2 $ and
	\[  \prob{ \| \hat{\theta_{\rm BF}}(n,\phi)-  \pmb q \|_\infty> \tfrac {\alpha n} 2  }\leq \tfrac \beta 2  \]
	The first of the two inequalities is an immediate consequence of a union bound along with the Hoeffding bound, stating that in the process of sampling the $n$ entries from the distribution, \begin{align*}
	\prob{ \exists j, |p_j- q_j| > \tfrac \alpha 2  }&\leq \sum_{j=1}^d \prob{|p_j- q_j| > \tfrac \alpha 2 }
	\cr &\leq  2d\exp(-2n\alpha^2/4) \leq \tfrac \beta 2
	\end{align*} since $n\geq  \tfrac 2 {\alpha^2} \log(\nicefrac {4d}\beta )$. Having fixed the input to have exactly $n\cdot q_j$ entries of each type $j$, it is evident that $\hat{\theta_{\rm BF}}$ is a function of the input $\phi(x_1),...,\phi(x_n)$ composed of $n$ standard basis vectors in $d$-dimensions. Our unbiased estimator thus satisfies that $\E[\hat{\theta_{\rm BF}}] = n\cdot \pmb q$, and moreover, each datum can affect the value of $\hat{\theta_{\rm BF}}$ by at most $\tfrac{e^{\epsilon/2}+1}{e^{\epsilon/2}-1}$. Applying a union bound along with McDiarmid's inequality, we get that 
	\begin{align*} &\prob{ \exists j, ~~ \left|\hat{\theta_{\rm BF}}_j- nq_j\right|> \tfrac {\alpha n} 2  }\leq \sum_{j=1}^d \prob{\left| \hat{\theta_{\rm BF}}_j- nq_j \right|> \tfrac {\alpha n} 2   } 
	\cr &~~~~~~\leq  2d \exp\left( -\tfrac {2\alpha^2 n^2 / 4} {n \cdot\left(\frac{e^{\epsilon/2}+1}{e^{\epsilon/2}-1}\right)^2}\right) 
	\cr &~~~~~~= 2d\exp\left( -n \frac{\alpha^2} {2\left( \tfrac{e^{\epsilon/2}+1}{e^{\epsilon/2}-1}\right)^2}  \right) \leq \tfrac \beta 2 \end{align*} as $n\geq \tfrac 2 {\alpha^2} \left( \tfrac{e^{\epsilon/2}+1}{e^{\epsilon/2}-1}\right)^2\log(\nicefrac {4d}\beta) $.
\end{proof}

\section{Confidence Intervals for the Mean with Known Variance}
\label{sec:known_variance}

In this section we assume that $\sigma$ is known and we want to estimate a confidence interval for $\mu$ based on a sample of $n$ users, subject to local differential privacy.  As in \citet{KV18}, we will break the algorithm into two parts.  First, we discretize the interval $[-R- \nicefrac \sigma 2,R+ \nicefrac \sigma 2]$ into bins of width $\sigma$, so that we have a collection of $d \stackrel{\rm def}{=} 2\lceil R/\sigma \rceil + 1$ disjoint intervals.
\begin{equation}
\cS(\sigma) = \cS_{-\lceil R/\sigma \rceil}(\sigma) \cup \cS_{-\lceil R/\sigma \rceil+1}(\sigma) \cup \cdots \cup \cS_{\lceil R/\sigma \rceil}(\sigma)
\label{eq:intervals}
\end{equation} where $\cS_i(\sigma) = [(i-\nicefrac 1 2)\cdot \sigma, (i+\nicefrac 1 2) \cdot \sigma]$. 
Denote $\phi:\mathbb{R}\to \{\pmb 0, \pmb e_1, \pmb e_2,..,\pmb e_d\}$ as the function that maps each $x$ to the indicating vector of the bin it resides in, and assigns any point outside the $[-R- \nicefrac \sigma 2,R+ \nicefrac \sigma 2]$ interval the all-$0$ vector, we can now apply the Bit Flipping mechanism to estimate  the histogram over the $d$ bins. Next, we find  the bin with the largest count, denoted $j^*$, and argue this bin is close up to two standard deviations to the true population mean $\mu$. We then move to the second part of the algorithm, where we place an interval $I$ of length $|I|=\tilde O(\sigma)$ around the $j^*$-th bin which is likely to hold all remaining points (a point outside this interval is projected onto the nearest point in $I$). Adding Gaussian noise to each point suffices to make the noisy result $(\epsilon,\delta)$-differentially private, and yet we can still sum over all points and obtain an estimation of the population mean which is close up to $\tilde O(\nicefrac \sigma {\sqrt n})$. Details are given in Algorithm~$\KnownBF$. We comment that we could replace the noise in the latter part by Laplace noise (rather than Gaussian) and obtain a $\epsilon$-LDP; this however would prevent us from (na\"ively) using the algorithm for the purpose of $Z$-test.
\begin{algorithm}
\caption{Known Variance Case: $\KnownBF$}
\label{alg:KnownBF}
\begin{algorithmic}[1]
\REQUIRE Data $\{ x_1, \cdots, x_{n}\}$; $\sigma$, $\beta$, $\epsilon$, $\delta$, $R$.
\STATE Set $n_1= 800\left( \frac{e^{\epsilon/2} +1}{e^{\epsilon/2}-1}\right)^2\log\left( \frac{8d}{\beta} \right)$ and $n_2= n-n_1$ where $d=2\lceil R /\sigma \rceil+1 $.
\STATE Partition the input into $\cU_1 = \{1,\cdots, n_1 \}$ and $\cU_2 = \{n_1+1, \cdots, n \}$.
\STATE Denote $\phi$ as the partition of the real-line into the $d$ bins as in \eqref{eq:intervals}.
\STATE Apply bit flipping on $\cU_1$: \mbox{$\tilde{\bp} \gets \tfrac 1 {n_1}\hat{\theta_{\rm BF}}(n_1,\phi)$} and let $j^*$ be the largest coordinate of $\tilde{\bp}$.
\STATE Set $\Delta = 2 \sigma + \sigma \sqrt{2 \log\left( \nicefrac{8n}{\beta} \right)}$. Denote the interval 
\begin{equation}
[s_1, s_2] = [j^*\sigma - \Delta, ~~j^*\sigma+\Delta]
\label{eq:end_points}
\end{equation}
and denote $\pi_{[s_1,s_2]}(x)=  \min \{ s_2, \max\{s_1,x\}  \}$, namely the projection of $x$ onto $[s_1,s_2]$.
\STATE Set $\hat\sigma^2 =  8\Delta^2 \log(\nicefrac 2 \delta) / \epsilon^2$.
\STATE \textbf{foreach} $i\in \cU_2$ \\
~~~~set $\tilde{x}_i = \pi_{[s_1,s_2]}(x_i) + N_i$ where $N_i \sim \Normal{0}{\hat{\sigma}^2 }$.
\STATE Set $\tilde \mu = \frac{1}{n_2}\sum\limits_{i \in \cU_2} \tilde{x}_i$
, and $\tau = \sqrt{\frac{\sigma^2+\hat{\sigma}^2}{n_2}} \cdot \Phi^{-1}(1-\nicefrac\beta{8})$
\ENSURE $I=[\tilde \mu - \tau, \tilde \mu+\tau]\cap [-R,R]$
\end{algorithmic}
\end{algorithm}

The following two theorems prove that Algorithm~$\KnownBF$ satisfies the required privacy and utility results.
\begin{theorem}
	\label{thm:known_var_is_epsilon_LDP}
$\KnownBF$ is $(\epsilon,\delta)$-LDP.
\end{theorem}
\begin{proof}
This follows from the fact that Algorithm~$\KnownBF$ applies one of two locally differentially private mechanisms to each datum~--- either bit flipping (which is known to be $\epsilon$-LDP) or additive random noise using Gaussian noise (a $(\epsilon,\delta)$-LDP algorithm).  
\end{proof}
\begin{theorem}
	\label{thm:known_varCI}
		Let $ \bX \stackrel{i.i.d.}{\sim} \Normal{\mu}{\sigma^2}$ and $I = \KnownBF\left( \bX; \sigma, \beta, \epsilon, \delta, n, R\right)$. Set \mbox{$d = 2\lceil\nicefrac {R}\sigma \rceil$ + 1}. If we have \mbox{$n \geq 1600 \left( \frac{e^{\epsilon/2} +1}{e^{\epsilon/2}-1}\right)^2 \log\left( \frac{8d}{\beta} \right)$}, then 
	$
	\Prob{\bX, \KnownBF}{\mu \in I} \geq 1- \beta
	$.
	Furthermore, 
	$$
	|I| = O \left( \sigma \cdot  \frac{\sqrt{\log\left(n / \beta \right)\cdot  \log\left(1/\beta \right) \cdot \log(1/\delta)  }   }{\epsilon \sqrt{n}} \right)
	$$

\end{theorem}

The utility analysis of our algorithm follows a similar analysis to Lemma 2.3 in \cite{KV18}. 
First note that Claim~\ref{clm:BF_utility_random_input} assures us that if $n\geq \tfrac 2 {\alpha^2} \left( \tfrac{e^\epsilon+1}{e^\epsilon-1}\right)^2\log(\tfrac 8 \beta)$, then each coordinate of $\tilde{\bp}$ is $\alpha$-close coordinate-wise to the true population histogram over the $d$ bins.
We show that for $n$ sufficiently large, selecting $j^*$ to be the largest coordinate of $\tilde{\bp}$ implies that we are close to $\mu$ within a constant multiple of the standard deviation $\sigma$.
\begin{lemma}
\label{lem:claim1}
Let $\{X_i\}_{i=1}^n \stackrel{i.i.d.}{\sim} \Normal{\mu}{ \sigma^2}$ and $\sigma$ be known and $\mu \in [-R,R]$.  Let $d = 2\lceil \frac{R}{\sigma} \rceil + 1$.  If $n \geq 800 \cdot \left( \frac{e^{\epsilon/2} +1}{e^{\epsilon/2}-1}\right)^2 \cdot \log\left( \frac{8d}{\beta} \right)$, then selecting $j^*$ as the largest coordinate of the histogram $\tilde{\bp}$ we have that w.p. $\geq 1-\nicefrac \beta 2$ the following holds
$$
|\mu - j^*\sigma | \leq 2\sigma
$$
\end{lemma}
\begin{proof}
The proof follows from the analysis done in Claim 1 of~\cite{KV18}.  We order the entries of the histogram $\tilde{\bp}$ in a non-ascending order as $\tilde{p}_{(1)} \geq \tilde{p}_{(2)} \geq  \cdots \geq \tilde{p}_{(d)}$.  We then have the following difference between the largest bin and the 3rd largest bin (note that the largest and second largest bin might have equal counts in the extreme case where the mean lies precisely between the two bins, but in any case the 3rd largest bin will be at least one standard deviation from the mean and must have noticeably smaller count)
$$
g = \tilde{p}_{(1)} - \tilde{p}_{(3)} \geq 0.1.
$$
If $\max_j\{ |\tilde{p}_j - p_j |\} \leq g/2$, then the index $j^*$ for the corresponding largest entry of $\tilde{\bp}$ will be within $2$ of the ratio $\mu/\sigma$.  Since each bin width is $\sigma$, we have $| j^* \sigma - \mu | \leq 2 \sigma$.
All that is left is to apply Claim~\ref{clm:BF_utility_random_input} with accuracy parameter $\alpha$ set as $g/2 =0.05$ and $\beta/2$.  This completes the proof.
\end{proof}

Next, conditioned on finding $j^*$ such that $|\mu-j^*\sigma|\leq 2\sigma$, we argue that the interval $[s_1,s_2]$ is sufficiently large so that w.h.p the projection $\pi_{[s_1,s_2]}$ onto this interval does not alter even a single one of the $n_2$ datapoints in $\cU_2$.
\begin{lemma}
\label{lem:claim3}
Suppose $j^*$ is an index satisfying the result of Lemma~\ref{lem:claim1}. Fix $n_2 \leq n$, and let $\{X_i\}_{i=1}^{n_2} \stackrel{i.i.d.}{\sim} \Normal{\mu}{\sigma^2}$. Then 
$$
\prob{\exists i \textrm{ s.t. }~ \{ |X_i - j^* \sigma |\} > \Delta  }   \leq  \nicefrac\beta 4
$$
\end{lemma}
\begin{proof}
We use the inequality $|X_i - j^* \sigma |  \leq | X_i - \mu| + | j^* \sigma - \mu|$. Lemma~\ref{lem:claim1} bounds $| j^* \sigma - \mu|~\leq~2\sigma$. Known concentration bounds for Gaussians give that $\prob{ |X_i-\mu| > \sigma\sqrt{2\log(\nicefrac{8n}{\beta})}  } ~\leq~\tfrac \beta {4n}$. Applying a union bound over $n_2 \leq n$ bad events concludes the proof.
\end{proof}

We can now provide the full utility analysis of Algorithm~$\KnownBF$. Namely, we argue that we indeed obtain a locally differentially private estimate for the mean of our data in the known variance case.  We advise the reader to compare this result to Theorem 4.1 in \cite{KV18} where the dependency on $\epsilon$ is (mainly) additive rather than multiplicative.

\begin{proof}[Proof of Theorem~\ref{thm:known_varCI}]
	Subject to Lemmas~\ref{lem:claim1} and~\ref{lem:claim3} holding, we have that w.p. $\geq 1- \tfrac{3\beta}4$ all of the latter $n_2$ datapoints in $\cU_2$ are not altered by $\pi_{[s_1,s_2]}$.  As each $X_i \sim\Normal{\mu}{\sigma^2}$ is added independent noise $N_i \sim\Normal{0}{\hat\sigma^2}$, conditioned on $\pi_{[s_1,s_2]}(X_i)=X_i$ we have that $\tilde x_i = \pi_{[s_1,s_2]}(X_i) + N_i = X_i + N_i \sim \Normal{\mu}{\sigma^2 + \hat{\sigma}^2}$. It thus follows that 
	\[ \tilde \mu \sim \Normal{\mu}{\tfrac{\sigma^2 + \hat{\sigma}^2} {n_2}} = \mu + \sqrt{\tfrac{\sigma^2 + \hat{\sigma}^2} {n_2}}\cdot \Normal{0}{1} \]
	By definition, we have that $\Prob{X\sim \Normal{0}{1}}{X > \Phi^{-1}(1-\nicefrac \beta 8)} = \nicefrac \beta 8$, and by the symmetry of the Gaussian {\sf PDF} we have that $\Prob{X\sim \Normal{0}{1}}{|X| > \Phi^{-1}(1-\nicefrac \beta 8)}=\nicefrac \beta 4$. Therefore, subject to Lemmas~\ref{lem:claim1} and~\ref{lem:claim3} holding, $\prob{ |\tilde \mu-\mu|  > \sqrt{\tfrac{\sigma^2 + \hat{\sigma}^2} {n_2}}\cdot \Phi^{-1}(1-\nicefrac \beta 8)} \leq \nicefrac \beta 4$. Thus we have that w.p. $\geq 1-\beta$ it holds that the output $\tilde \mu$ of our algorithm satisfies ${ |\tilde \mu-\mu|  > \sqrt{\tfrac{\sigma^2 + \hat{\sigma}^2} {n_2}}\cdot \Phi^{-1}(1-\nicefrac \beta 8)}$ proving the first part of the theorem.
	
	The second part of the theorem follows for standard bounds on the Normal distribution, we state that $\Phi^{-1}(1-\tfrac \beta 8) \leq \sqrt{2\log(\tfrac{2\cdot 8}\beta)} = O(\sqrt{\log(\nicefrac 1 \beta)})$. The remainder follows from the definition of $\hat{\sigma}$ and $\Delta$ in Algorithm~\ref{alg:KnownBF}, and the fact that when $n\geq 2n_1$ then $n_2 = n-n_1 \geq n/2$.
\end{proof}

We can now apply Theorem~\ref{thm:approx2pure}, where we pick $\beta = \eta$ in the Theorem statement, along with Theorem~\ref{thm:known_varCI} to obtain a valid confidence interval subject to pure $\epsilon$-DP.

\begin{corollary}
Fix $\epsilon < 1$, set $d = 2\lceil\nicefrac {R}\sigma \rceil + 1 $, and let $n \geq 1600 \left( \frac{e^{\epsilon/20} +1}{e^{\epsilon/20}-1}\right)^2 \log\left( \frac{16d}{\beta} \right)$.  There exists an algorithm that returns a valid $1-\beta$-confidence interval $I$ that is $\epsilon$-LDP and 
$$
|I| = O\left( \sigma \cdot \frac{\sqrt{\log(n/\beta) \cdot \log(1/\beta) \cdot \log\left( \frac{n \log(n/\beta)}{\epsilon \beta} \right) }}{\epsilon \sqrt{n}} \right)
$$
\end{corollary}


\subsection{Experiment: $Z$-Test}
\label{subsec:Ztest}

As in Algorithm~$\KnownBF$, we denote $n_1  \stackrel{\rm def}{=}800 \cdot \left( \frac{e^{\epsilon/2} +1}{e^{\epsilon/2}-1}\right)^2 \cdot \log\left( \frac{8d}{\beta} \right)$ and $n_2  \stackrel{\rm def}{=} n-n_1$. Following the proof of Theorem~\ref{thm:known_varCI}, we have that --- under the assumption that no datapoint is clipped --- all $n_2$ datapoints we use in the latter part of Algorithm~\ref{alg:KnownBF}  are sampled from $\Normal{\mu}{\sigma^2 + \hat{\sigma}^2}$. This allows us to infer that (w.p. $\geq 1-\beta$) the average of the $n_2$ datapoints in $\cU_2$ is sampled from $\Normal{\mu}{\frac{\sigma^2 + \hat{\sigma}^2}{n_2}}$. Just as in Algorithm~\ref{alg:KnownBF}, denoting $\tilde \mu$ as the average of the noisy datapoints, we now can define an approximation of the likelihood: $\calP = \Normal{\tilde \mu}{\frac{\sigma^2 + \hat{\sigma}^2}{n_2}}$. As a result, for any interval on the reals $I$ we can associate a likelihood of $p_I \stackrel{\rm def}= \Prob{X\sim \calP}{X\in I}$, and we know that w.p. $p_I \pm \beta$ it indeed holds that $\mu\in I$. This mimics the power of a $Z$-test~\citep{HoggMC05} ~--- in particular we can now compare two intervals as to which one is more likely to hold $\mu$, compare populations, etc.

Note however that, as opposed to standard $Z$-test, the result of Algorithm~\ref{alg:KnownBF} only gives confidence bounds up to an error of $\beta$. So for example, given two intervals $I$ and $I'$ we can safely argue that it is more likely that $\mu\in I$ than $\mu\in I'$ only when $p_I > p_{I'}+2\beta$. Similarly, if we wish to draw an interval whose likelihood to contain $\mu$ is $1-\nu$ for some $\nu>0$, we must pick a corresponding $(1-\nu+\beta)$-confidence interval from $\calP$. Naturally, this limits us to the setting where $\beta < \nu$, or conversely: we can never allow for more certainty than the $1-\beta$ parameter specified as an input for Algorithm~\ref{alg:KnownBF}.

Subject to this caveat, Algorithm~\ref{alg:KnownBF} allows us to perform $Z$-test in a similar fashion to the standard $Z$-test, after we omit the first $n_1$ datapoints from our sample. One of the more common uses of $Z$-test is to test whether a given sample behaves in a similar fashion to the general population. For example, suppose that  the SAT scores of the entire population are distributed like a Gaussian of mean $\mu$ and variance $\sigma^2$. Taking a sample of SAT scores from one specific city, we can apply the $Z$-test to see if we can reject the null hypothesis that the score distribution in this city are distributed just as they are distributed in the general population. Should we have $n$ samples of SAT scores which happen to be distributed from $\Normal{\mu'}{\sigma^2}$ for some $\mu'\neq \mu$, then sufficiently large $n$ (with dependency on $|\mu'-\mu|$) should allow us to reject this null hypothesis with confidence $1-\nu$. We set to discover precisely this notion of utility, using our locally-private $Z$-test.

\noindent\emph{The Experiment:} We tested our LDP $Z$-test on $n$ iid samples from a Gaussian. We set the null-hypothesis to be $H_0: \Normal{0}{1}$, whereas the $n$ samples were drawn from the alternative hypothesis $H_1: \Normal{\mu'}{1}$ with $\mu'>0$. We run our experiments in the known variance $\sigma^2 = 1$ case with a fixed bound $R=200$ and $\beta = 0.01$.  In each set of experiments we vary $\epsilon$ while keeping $\delta = 10^{-9}$.  In Figure~\ref{fig:pvalues}, we plot the average p-value over 1,000 trails for our Z-test when the data is actually generated with sample size $n = 200,000$ and mean $\mu'$ that varies.  In Figure~\ref{fig:power}, we plot the empirical power of our test over 
1000 
trails where we fix $\mu' = 3$ and vary the sample size $n$.  Our figures show the tradeoffs between the privacy parameter, the alternate we are comparing the null to, and the sample size.  
The results themselves match the theory pretty well and emphasize the magnitude of the needed sample size. For $\epsilon=1.5$ we need 10,000 sample points to reject the null hypothesis w.h.p. When $\epsilon=0.5$,  even 100,000 sample points do not suffice to reject the null hypothesis w.h.p despite the fact that the difference between the means of the null and the alternative is $3$ times greater than the variance. This is a setting where non-privately we can reject the null hypothesis with a sample size $<100$. This illustrates (yet again) how LDP relies on the abundance of data.

\begin{figure}[t!]
    \centering
    \begin{subfigure}[t]{0.49\textwidth}
        \centering
        \includegraphics[height=1.55in]{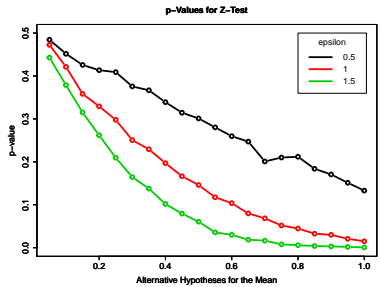}
        \caption{Average p-values with $n = 200,000$.\label{fig:pvalues}}
    \end{subfigure}
    \begin{subfigure}[t]{0.49\textwidth}
        \centering
        \includegraphics[height=1.55in]{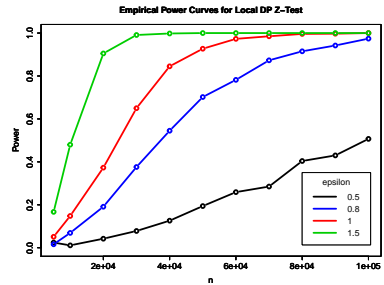}
        \caption{Empirical power with alternate $\mu' = 3$.\label{fig:power}}
    \end{subfigure}
    \caption{Z-test experiments showing the empirical p-values and power averaged over 100 trials for various privacy parameters.\label{fig:experiments}}
    \vspace{-10pt}
\end{figure}

\section{Mean Estimation with Unknown (Bounded) Variance}
\label{sec:unknown_variance_algorithm}

In this section we discuss the problem of locally private mean estimation in the case where the variance of the underlying population is unknown. For ease of exposition, we separate this case into two sub-cases. First, we assume that the variance is bounded by some $\sigma_{\max} \leq 2R$ and it is the sole focus of this section as it the more likely of the two. In the second case, we consider very-large variance ($\sigma>R$), a case which~\citet{KV18} do not analyze, and it is deferred to Section~\ref{sec:large_unknown_variance}. As our lower bounds show, our algorithm must be provided bounds $\sigma_{\min}>0$ and $\sigma_{\max}\leq 2R$ such that $\sigma \in [\sigma_{\min},\sigma_{\max}]$. As we show, our parameters dependency on these upper- and lower-bounds on the variance is logarithmic (so, for example, $\sigma_{\min} > \nicefrac 1 {R^2}$ is a useful bound for us).

Our overall approach in this section mimics the same approach from Algorithm~\ref{alg:KnownBF}. Our goal is to find a suitably large, yet sufficiently tight interval $[s_1,s_2]$ that is likely to hold the latter part of the input. However, finding this $[s_1,s_2]$-interval cannot be done using the off-the-shelf bit flipping mechanism as that requires that we know the granularity of each bin in advance. Indeed, if we discretize the interval $[-R,R]$  with an upper-bound on the variance, each bin might be far too large and result in an interval $[s_1,s_2]$ which is far larger than the variance of the underlying population; and if we were to discretize $[-R,R]$ with a lower-bound on the variance we cannot guarantee substantial differences between the bins that are close to $\mu$. And so, we abandon the idea of finding a histogram on the data. Instead, we propose finding a good approximation for $\sigma$ via quantile estimation based on a binary search. This result is likely to be of independent interest. Once we establish formal guarantees on our locally private binary search algorithm (privacy and utility bounds), we plug those into our confidence interval estimation algorithm in Subsection~\ref{subsec:algorithm_unknown_variance}.

\subsection{Locally Private Binary Search and Quantile Estimation}
\label{subsec:LDP_binary_search}
\newcommand{\adist}{\ensuremath{\alpha_{\rm dist}}}
\newcommand{\aquant}{\ensuremath{\alpha_{\rm quant}}}

We now show how to estimate quantiles of a probability distribution using randomized response and binary search. We assume our domain $\cX$ is contained in the real line and that there exists some distribution $\calP$ over this domain. We define the quantile $t$ as $p(t) = \Prob{\calP}{X < t}$.  Given a target probability $p^*$, let $t^*$ be the quantile we want to estimate, namely $p(t^*)=p^*$.  We will say that $t^*$ is a $p^*$-quantile of $\cP$ when $p^*=\Prob{x\sim \calP}{x\leq t^*}$.
Since our algorithm is randomized and therefore uses only estimations, we must allow for some error $\lambda$, and find some $t$ such that $|p(t)-p^*| \leq \lambda$ with high probability.

Our binary search begins with some bounded interval guaranteed to contain $t^*$, i.e. $t^* \in [Q_{\min}, Q_{\max}]$.   Initially, we set $t^{(0)} = \frac{Q_{\max} + Q_{\min}}{2}$, and draw a subsample of size $m$, where $m$ is chosen so that w.h.p. we can estimate $\E_{X\sim\calP}[\1\{X < t^{(0)}\}]$ using randomized response up to an error of $\lambda$. 
Denoting the randomized response estimator as $\hat{\theta_{\rm RR}}^{(0)}$
one of the following three must holds. Either (i) $|\hat{\theta_{\rm RR}}^{(0)} - p^*| \leq \lambda$, in which case we have found a good enough approximation for $t^*$ and we may halt; or (ii) $\hat{\theta_{\rm RR}}^{(0)} > p^*+\lambda $ in which case $t^{(0)}$ is too large, and so $t^* \in [Q_{\min},t^{(0)}]$ and we recurse of the LHS half of the original interval; or (iii) $\hat{\theta_{\rm RR}}^{(0)} < p^*-\lambda $ in which case $t^{(0)}$ is too small, and so $t^* \in [t^{(0)},Q_{\max}]$  and we recurse on the RHS half of the original interval.

When does our binary search algorithm halt? If $\calP$ is a pathological distribution, it may put $2\lambda$ probability mass on an infinitesimally small intervals to the left and right of $t^*$, forcing our binary search algorithm to continue for arbitrarily many rounds. To avoid such a case, we require an a-priori bound $\adist$ on the length of an interval that can hold $\lambda$-probability mass; or alternatively, allow our algorithm to output any $t$ such that $|t-t^*|\leq \adist$. The formal definition follows.
\begin{definition}
	\label{def:solving_quantiles}
	Let $t^*$ be the $p^*$-quantile of $\calP$ and assume that $t^*$ is bounded, i.e. $t^*\in [Q_{\min},Q_{\max}]$.  An algorithm $\calM$ is said to \emph{$(\adist,\aquant,\beta)$-approximate} $t^*$, if it takes as input $n$ iid draws from $\calP$ and returns $t\in[Q_{\min},Q_{\max}]$ such that w.p. $\geq 1-\beta$ we have that either $|p^*-\Prob{x\sim \calP}{x\leq t}|\leq\aquant$ or that $|t-t^*| \leq \adist$.
\end{definition}
Provided with such a bound $\adist$ we can bound the number of iterations in our binary search by $T$ such that $\nicefrac{Q_{\max}-Q_{\min}} {2^T} < \adist$.  A description of our binary search given such an iteration bound $T$ is detailed in procedure $\BinRR$ given in Algorithm~\ref{alg:BinaryQuantile}.

\begin{algorithm}
\caption{Quantile Estimation: $\BinRR$}
\label{alg:BinaryQuantile}
\begin{algorithmic}
\REQUIRE Data $\{ x_1, \cdots, x_N\}$, target quantile $p^*$; $\epsilon$, $[Q_{\min},Q_{\max}]$, $\lambda$, $T$.
\STATE Initialize $j = 0$, $n = N/T$,  $s_1=Q_{\min}, s_2=Q_{\max}$.
\FOR{$j = 1, \cdots, T$}
	\STATE Select users $\cU^{(j)} = \{j\cdot n+1, j\cdot n +2, \cdots, (j+1)\cdot n \}$
	\STATE Set $ t^{(j)} \gets \frac{s_1+s_2 }{2}$ 
	\STATE Denote $\phi^{(j)}(x) = \1\{ x< t^{(j)}  \}$.
	\STATE Run randomized response on $\cU^{(j)}$ and obtain $Z^{(j)} = \tfrac 1 n\hat{\theta_{\rm RR}}(n, \phi^{(j)})$. 
	\IF{($Z^{(j)}> p^* + \tfrac \lambda 2$)}
		\STATE $s_2 \gets t^{(j)}$
 	\ELSIF{(${Z}^{(j)} < p^*- \tfrac \lambda 2$)}
		\STATE $s_1 \gets t^{(j)}$
	\ELSE \STATE break
 	\ENDIF
\ENDFOR
\ENSURE $t^{(j)}$
\end{algorithmic}
\end{algorithm}

Two theorems summarize Algorithm~\ref{alg:BinaryQuantile}'s properties. 
\begin{theorem}
$\BinRR$ is $\epsilon$-LDP.
\end{theorem}
\begin{proof}
This follows immediately from the fact that the only time we access the data is via randomized response, which is $\epsilon$-DP.
\end{proof}

\DeclareRobustCommand{\thmBinarySearchUtility}
{Let $\calP$ be any distribution on the real line. For any $p^*\in (0,1)$ and any $Q_{\min},Q_{\max}$ such that $q^* \in [Q_{\min},Q_{max}]$, for any $\epsilon>0$ and for any $\lambda,\tau,\beta \in (0,\nicefrac 1 2)$, Algorithm $\BinRR$ indeed $(\tau,\lambda,\beta)$-approximates the $p^*$-quantile if  $T = \lceil\log_2(\frac{Q_{\max}-Q_{\min}}{\tau})\rceil$ with $N$ iid draws from $\calP$, provided that $N\geq \tfrac {8T} {\lambda^2} \left(\frac{e^\epsilon+1}{e^\epsilon-1}\right)^2 \log(\nicefrac {4T}\beta)$.}
\begin{theorem}
	\label{thm:utility_quantile}
	\thmBinarySearchUtility
\end{theorem}
\begin{proof}
	We know our algorithm applies the randomized response mechanism at most $T$ times. Setting the probability of each iteration $j$ to produce a bad estimation as $\nicefrac \beta  T$, Claim~\ref{clm:BF_utility_random_input} assures us that it suffices to run each iteration on $\tfrac 8 {\lambda^2} \left(\frac{e^\epsilon+1}{e^\epsilon-1}\right)^2 \log(\nicefrac {4T}\beta)$ samples to have that each $Z^{(j)}$ approximates $\Prob{X\sim\calP}{X<t^{(j)}}$ up to an error of $\tfrac \lambda 2$. Taking a union bound over all $T$ iterations, we have that w.p. $\geq 1-\beta$, the condition $ |Z^{(j)} - \Prob{X\sim\calP}{X<t^{(j)}}| \leq \tfrac{\lambda}{2}$ holds for every $j$. And so, if we have that $Z^{(j)} > p^*+\tfrac \lambda 2$ then it must hold that $p^* < \Prob{X\sim\calP}{X<t^{(j)}}$ which implies we must increase $t^{(j)}$; and if we have that $Z^{(j)} < p^*-\tfrac \lambda 2$ then it must hold that $p^* > \Prob{X\sim\calP}{X<t^{(j)}}$ which implies we must decrease $t^{(j)}$. Thus, we maintain the invariant that in each iteration $t^* \in [s_1,s_2]$.
	
	If our algorithm breaks at some iteration, it implies that the estimate $Z^{(j)}$ of that particular iteration is sufficiently close to $p^*$, thus $|\Prob{X\sim\calP}{X<t^{(j)}} - p^*|\leq \lambda$. Otherwise, we halt after $T$ iterations, which means that difference $s_2-s_1$, which initially was set to $Q_{\max}-Q_{\min}$ is cut in half $T$ times. Thus, after $T$ iteration we have that $s_2-s_1 = \frac{Q_{\max}-Q_{\min}}{2^T} \leq \tau$.
\end{proof}

\subsection{Locally Private Mean Estimation Using Quantile Estimation}
\label{subsec:algorithm_unknown_variance}

We return to discuss the case where the underlying distribution of the data is Gaussian with unknown variance. Recall, our plan is to use quantile estimation to find an interval $[s_1,s_2]$ which is likely to contain most datapoints. This requires that we assess $\mu$ up to an error of about $\pm\sigma$ and also have an estimation of $\sigma$ which is also fairly close to the true $\sigma$. In other words, by denoting $\tilde \sigma$ as our estimation, we would like to have $\tfrac \sigma 2 \leq \tilde \sigma \leq 2\sigma$.

Our approach for obtaining such estimations of $\mu$ and $\sigma$ is to apply the quantile estimation technique \emph{twice}: once for $p^*=\tfrac 1 2$ where $t^*=\mu$, and once for the value of $p^*= \Phi(1) \approx 0.8413$ for which the corresponding threshold is $t^*=\mu+\sigma$.
We next argue that since both thresholds are sufficiently close to the mean of the underlying distribution, we can set $\lambda$ as a reasonable constant and guarantee that our estimations of the two thresholds are close up to a factor of $\sigma/4$ to the true thresholds. 
Our LDP confidence interval estimator in the unknown variance case is given in Algorithm~\ref{alg:UnknownVar}.

Recall that we assume that $\mu \in [-R,R]$ and $\sigma \in [\sigma_{\min}, \sigma_{\max}]$, hence we can set our $\adist$ parameter to be $\sigma_{\min}/4$ and $T = \lceil\log_2(\nicefrac{2R} {\adist})\rceil$.  We start by using $\BinRR$ to estimate the mean and the threshold $\mu+\sigma$. We leverage on Theorem~\ref{thm:utility_quantile} to show the following.
\begin{corollary}
\label{lem:quantile_mean}
Fix any $\lambda \in (0, 0.1)$. Let $\bX \sim \Normal{\mu}{\sigma^2}$ be iid samples. Set 
\begin{align}
&\alpha^{\rm med}=\Phi^{-1}(\nicefrac 1 2 + \lambda) - \Phi^{-1}(\nicefrac 1 2) = \Phi^{-1}(\nicefrac 1 2 + \lambda)
\cr &\textrm{ and } T^{\rm med}=\lceil\log_2(\nicefrac{2R} {\alpha \cdot \sigma_{\min}})\rceil
\label{eq:T_mu}
\end{align}
and compute the estimate 
\begin{equation}
\hat{t}_\mu = {\BinRR\Big(\{X_1, \cdots, X_{n}\}, p^*=\tfrac 1 2; \epsilon, [-R,R],\lambda, T^{\rm med}\Big)}.
\label{eq:q_mu}
\end{equation}
If $n \geq \tfrac{T^{\rm med}}{\lambda^2} \cdot \left( \frac{e^\epsilon +1}{e^\epsilon -1 } \right)^2 \cdot \log(\nicefrac {8T^{\rm med}} \beta)$
then with probability $\geq 1-\tfrac \beta 2$ it holds that $ |\hat{t}_\mu - \mu| \leq \alpha^{\rm med} \cdot \sigma$.

Similarly, set
\begin{align}
&\alpha^{\rm sd} = \Phi^{-1}( \Phi(1) + \lambda)  - 1 \cr & \textrm{ and }
  T^{\rm sd}=\lceil\log_2(\nicefrac{2R+\sigma_{\max}} {\alpha \cdot \sigma_{\min}})\rceil
\label{eq:T_sigma}
\end{align}
and compute the estimate 
\begin{equation}
\hat{t}_\sigma = \BinRR(\{X_1, \cdots, X_{n}\}, p^*=\Phi(1); \epsilon, [-R, R + \sigma_{\max}],\lambda,T^{\rm sd}).
\label{eq:q_sigma}
\end{equation}
If $n \geq \tfrac{T^{\rm sd}}{\lambda^2} \cdot \left( \frac{e^\epsilon +1}{e^\epsilon -1 } \right)^2 \cdot \log(\nicefrac {8T^{\rm sd}} \beta)$
then with probability $\geq 1-\tfrac \beta 2$ it holds that $ |\hat{t}_\sigma - (\mu+\sigma)| \leq \alpha^{\rm sd} \cdot \sigma$.
\end{corollary}
\begin{proof}
We begin with the estimation $\hat t_\mu$. Since we have that $\mu \in [-R,R]$ and since our bound on $n$ meets the requirements of Theorem~\ref{thm:utility_quantile}, then w.p. $\geq 1-\tfrac \beta 2$ our algorithm succeeds and returns $\hat t _\mu$ such that $p(\hat{t}_\mu)$ is within $\pm\lambda$ of $p(\mu) =1/2$, then we know that $\left|\Phi\left( \frac{\hat{t}_\mu - \mu}{\sigma} \right) - 1/2\right| \leq \lambda$, which implies that
\begin{equation}
 \hat{t}_\mu \in \left[\mu + \Phi^{-1}(\nicefrac 1 2-\lambda) \cdot\sigma , \mu + \Phi^{-1}(\nicefrac 1 2 +\lambda) \cdot \sigma  \right].
\label{eq:meanRR}
\end{equation}
Leveraging the symmetry of the Normal distribution around the mean (hence the mean is also the median), we know that $t_L = \Phi^{-1}(\nicefrac 1 2-\lambda)$ is precisely $-t_U$ with $t_U = \Phi^{-1}(\nicefrac 1 2 + \lambda)=\alpha^{\rm med}$. Thus we have that $\hat t_\mu \in [\mu - \alpha^{\rm med} \sigma, \mu + \alpha^{\rm med}\sigma]$.

We now make a similar argument for the estimation $\hat t_{\sigma}$. Since we have that $\mu+\sigma \in [-R,R+\sigma_{\max}]$ and since our bound on $n$ meets the requirements of Theorem~\ref{thm:utility_quantile}, then w.p. $\geq 1-\tfrac \beta 2$ our algorithm succeeds and returns $\hat t _\mu$ such that $p(\hat{t}_\mu)$ is within $\pm\lambda$ of $p(\mu) =1/2$, then we know that $\left|\Phi\left( \frac{\hat{t}_\sigma - \mu}{\sigma} \right) - \Phi\left( 1 \right)\right| \leq \lambda$ which implies that
\begin{equation}
\hat{t}_\sigma  \in \left[\mu + \Phi^{-1}\left( \Phi(1) - \lambda \right) \cdot \sigma ,~~ \mu + \Phi^{-1}\left( \Phi(1) + \lambda \right) \cdot \sigma\right].
\label{eq:sigmaRR}
\end{equation}
We re-write the endpoints of the interval as
\begin{align*}
\mu + \Phi^{-1}\left( \Phi(1) - \lambda \right) \cdot \sigma & = \mu + \sigma + \left(\Phi^{-1}\left( \Phi(1) - \lambda \right) -1 \right) \sigma \\
& = \mu + \sigma - \left(\Phi^{-1}(\Phi(1)) - \Phi^{-1}\left( \Phi(1) - \lambda \right)  \right) \sigma
\\
\mu + \Phi^{-1}\left( \Phi(1) + \lambda \right) \cdot \sigma & = \mu + \sigma + \left(\Phi^{-1}\left( \Phi(1) + \lambda \right) -1 \right) \sigma  \\
& = \mu + \sigma + \left(\Phi^{-1}\left( \Phi(1) + \lambda \right) - \Phi^{-1}(\Phi(1)) \right) \sigma
\end{align*}
We denote $\alpha_{-} = 1-\Phi^{-1}(\Phi(1-\lambda))$ and $\alpha^{\rm sd} = \Phi^{-1}(\Phi(1+\lambda))-1 $.
Recall we limit $\lambda < 0.1$ and so due to the monotonically decreasing nature of the {\sf PDF} of the Normal distribution on the interval $(0.9, 1.1)$ it must hold that $\alpha^{\rm sd} > \alpha_{-}$ as the interval $(1, 1+\alpha^{\rm sd})$ must be longer than the interval $(1-\alpha_{-}, 1)$ in order to accumulate $\lambda$ probability mass. Therefore, we have that $|\hat t_\sigma-(\mu+\sigma)| \leq \alpha^{\rm sd}$. 
\end{proof}

All the is now left is to ``plug in'' the value of $\lambda$ for which $\alpha^{\rm med} = \nicefrac 1 4$ and the value of $\lambda$ for which $\alpha^{\rm sd} = \nicefrac 1 4$. Referring to known quantile calculations over the standard Normal, it is enough to set $\lambda = 0.098$ to have $\alpha^{\rm med} \leq \nicefrac 1 4$, and to set $\lambda = 0.052$ to have $\alpha^{\rm sd}\leq \nicefrac 1 4$. Under these values, Corollary~\ref{lem:quantile_mean} asserts that w.p. $\geq 1 - \beta$ we have $|\hat t_\mu - \mu| \leq \tfrac \sigma 4$ and also $\tfrac \sigma 2 \leq |\hat t_\mu - \hat t_{\sigma}| \leq \tfrac {3\sigma} 2$.
This allows us to follow in a similar fashion to Algorithm~\KnownBF, i.e. to define a suitably chosen interval centered at $\hat t_\mu$ which is wide enough to hold all remaining data points yet its length is still $O(\sigma \sqrt{\log(2/\beta)})$. Details appear in Algorithm~\ref{alg:UnknownVar}.

\begin{algorithm}
\caption{Unknown Variance Case: $\UnkVar$}
\label{alg:UnknownVar}
\begin{algorithmic}
\REQUIRE Data $\{ x_1, \cdots, x_N\}$;$\lambda$, $\sigma_{\min},\sigma_{\max}$, $\beta$, $\epsilon$, $\delta$, $R$
\STATE Set $T^{\rm med} = \left\lceil\log_2(\frac{8R} {\sigma_{\min}})\right\rceil$,$T^{\rm sd} = \left\lceil\log_2(\frac{8R+4\sigma_{\max}} {\sigma_{\min}}) \right\rceil$.
\STATE Set $n_1 = \tfrac{T^{\rm med}}{(0.098)^2} \cdot \left( \frac{e^\epsilon +1}{e^\epsilon -1 } \right)^2 \cdot \log(\nicefrac {16T^{\rm med}} \beta)$, $n_2 = \tfrac{T^{\rm sd}}{(0.052)^2} \cdot \left( \frac{e^\epsilon +1}{e^\epsilon -1 } \right)^2 \cdot \log(\nicefrac {16T^{\rm sd}} \beta)$, and $n_3 = n-n_1-n_2$.
\STATE Initialize $\cU_1 = \{1, \cdots, n_1 \}$, $\cU_2 = \{n_1+1, \cdots, n_1+n_2\}$, and $\cU_3 = \{n_1+n_2 + 1, \cdots, n \} $.
\STATE $\hat{t}_\mu \gets \BinRR(\{x_i : i \in \cU_1\}, 1/2;~~\epsilon,n, [-R, R], 0.098, T^{\rm med})$
\STATE $\hat{t}_\sigma \gets \BinRR(\{x_i : i \in \cU_2\}, \Phi(1); \epsilon,n, [-R, R+\sigma_{\max}], 0.052,T^{\rm sd})$ 
\STATE Set $\Delta = (\hat t_\sigma - \hat t_\mu)\cdot (\tfrac 1 2 + 2\sqrt{2\log(\nicefrac {8n}{\beta})})$
\STATE Denote the interval $[s_1,s_2] = [\hat t_\mu - \Delta, \hat t_\mu + \Delta]$.
\STATE Run steps $6$-$8$ of Algorithm~$\KnownBF$ over $\cU_3$ to get the interval $\hat{I}$.
\ENSURE $\hat{I}$
\end{algorithmic}
\end{algorithm}
\DeclareRobustCommand{\thmUtilityUnknownVar}{Let $\bX \stackrel{i.i.d.}{\sim} \Normal{\mu}{\sigma^2}$.  Fix parameters $\epsilon$, $\beta \in (0,\nicefrac 1 2)$. Given that $\mu \in [-R,R]$ and that $\sigma_{\min}\leq\sigma \leq \sigma_{\max}\leq 2R$, if
	$$
	n \geq { 1500\log_2(\tfrac{16R}{\sigma_{\min}}) } \cdot \left( \tfrac{e^\epsilon +1}{e^\epsilon -1 } \right)^2 \cdot \log(\tfrac {16\log_2(\nicefrac{16R}{\sigma_{\min}})} \beta)
	$$
	then the interval $\hat I$ returned by Algorithm~$\UnkVar$ satisfies that  $\Prob{\bX,~ \UnkVar}{\mu \in \hat{I}  } \geq  1- \beta$, and moreover 
	\[\hat{I}  = O \left( \sigma \cdot  \frac{\sqrt{\log\left(n / \beta \right) \log\left(1/\beta \right)\log(1/\delta)}}{\epsilon \sqrt{n}} \right). \]}
\begin{theorem}
\label{thm:utility_unknown}
\thmUtilityUnknownVar
\end{theorem}
\begin{proof}
Denoting $T^{\rm sd}$ as in Algorithm~\ref{alg:UnknownVar}, we have set this particular value of $n$ so that $n\geq 3n_2 \geq 3n_1$, implying $n_3 \geq n/3$. Using Corollary~\ref{lem:quantile_mean}, we know that under these particular values of $\lambda$ (namely $0.098$ and $0.052$) we have that w.p. $\geq 1-\tfrac{\beta}{2}$ both $|\hat t_\mu -\mu |\leq \tfrac \sigma 4$ and $|\hat t_\sigma - (\mu+\sigma)|| \leq \tfrac \sigma 4$, and as a result $\sigma \leq 2(\hat t_\sigma - \hat t_\mu)\leq 3\sigma$. Therefore, we have that the interval $[ \mu - \sigma\sqrt{2\log(\nicefrac{8n}\beta)},\mu + \sigma\sqrt{2\log(\nicefrac{8n}\beta)} ] \subset [s_1,s_2]$. As discussed in the proof of Theorem~\ref{thm:known_varCI}, this implies that w.p. $\geq 1-\tfrac{\beta}{4}$ none of the points of $\cU_3$ is altered by the projection onto $[s_1,s_2]$. The remainder of the argument then follows in the same fashion. As for the length of the resulting interval, we have that the length of $[s_1,s_2] = 2\Delta$ is also upper bounded by $6\sigma(\tfrac 1 2 + 2\sqrt{2\log(\nicefrac {8n}\beta)})$, making the standard deviation of the Gaussian noise we add to each point proportional to $O \left( \frac\sigma\epsilon \cdot  {\sqrt{\log\left(n / \beta \right) \log\left(1/\beta \right)\log(1/\delta)}} \right)$.
\end{proof}

It is interesting to compare the bounds of Theorems~\ref{thm:known_varCI} and~\ref{thm:utility_unknown}. Replacing the known quantity $\sigma$ in Theorem~\ref{thm:known_varCI} with the provided lower bound $\sigma_{\min}$ in Theorem~\ref{thm:utility_unknown}, the sample complexity bound only increases by a $\log\log(R/\sigma_{\min})$-factor. Note in both algorithms we conclude in a similar fashion (averaging Gaussian noise), so, if we are to denote by $m$ the number of points either algorithms use in their last parts, then both algorithms output intervals of length $\tilde O(\nicefrac \sigma {\epsilon\sqrt m})$.

\section{Dealing with Very Large Unknown Variance}
\label{sec:large_unknown_variance}

We now move to dealing with the case where the true variance of the data is at least as large as $R$. We heavily rely on some known properties regarding quantiles of the Gaussian distribution. Though cumbersome to state, we use a claim that argues shifting a threshold $t$ to a new threshold $t'=t+\Delta$, the difference between the probability mass for a Standard Gaussian distribution, i.e. $|\Phi(t)-\Phi(t')|= \Theta(\Delta)$, provided the threshold $t$ isn't too far from the mean of the Gaussian (in our case, at most two standard deviations away).

\begin{claim}
\label{clm:perturbations_of_quantiles}
Let $\calN(\mu,\sigma^2)$ be the Gaussian distribution set at mean $\mu$ and with variance $\sigma^2$. 
\begin{enumerate}
\item 
Fix $\alpha,p\in (0,1)$ such that $p-\alpha \geq 2.28\%$ and $p+\alpha\leq 97.72\%$. Let $t$ be the $p$-quantile of $\calN(\mu,\sigma^2)$, i.e $\Prob{X\sim\calN(\mu,\sigma^2)}{X<t}=p$, and let $\tilde t$ be a $\alpha$-approximation of $t$, i.e. some arbitrary point satisfying $\Pr_{x\sim\calN(\mu,\sigma^2)}[x<\tilde t]\in (p-\alpha,p+\alpha)$. Then $ |t-\tilde t| \leq 18.522\alpha\sigma$.
\item 
Fix $\tau,p\in (0,1)$, and let $t$ be the $p$-quantile of $\calN(\mu,\sigma^2)$, i.e $p=\Prob{X\sim\calN(\mu,\sigma^2)}{X<t}$. Let $\tilde t$ be any point such that $|\tilde t - t| < \tau\sigma$. Then $\left| \Pr_{x\sim\calN(\mu,\sigma^2)}[x<\tilde t] - p\right| \leq 0.399\tau$. 
\end{enumerate}
\end{claim}
\begin{proof}
The proof is nothing more than a few calculations using the mean value theorem.
\begin{enumerate}
\item 
Without loss of generality, assume $\tilde t < t$. Thus $\alpha \geq \int\limits_{\tilde t}^t (2\pi\sigma^2)^{-0.5} \exp(-\tfrac {(x-\mu)^2}{2\sigma^2})dx = |t-\tilde t| \cdot  (2\pi\sigma^2)^{-0.5} \exp(-\tfrac {(c-\mu)^2}{2\sigma^2})$ for some $c\in [\tilde t,t]$. Since $p-\alpha > 0.0228$ then $\tilde t > \mu-2\sigma$, and since $t+\alpha<0.9772$ then $t < \mu+2\sigma$; implying that $\exp(-\tfrac {(c-\mu)^2}{2\sigma^2}) > e^{-\frac {2^2} 2}$. Thus $|t-\tilde t| < \alpha\sigma(2\pi)^{0.5}e^2 < 18.522\alpha\sigma$.
\item Again, without loss of generality we assume $\tilde t < t$, and thus $\left| \Pr_{x\sim\calN(\mu,\sigma^2)}[x<\tilde t] - p\right| = \int\limits_{\tilde t}^t (2\pi\sigma^2)^{-0.5} \exp(-\tfrac {(x-\mu)^2}{2\sigma^2}) dx= |\tilde t -  t|\cdot (2\pi\sigma^2)^{-0.5}\exp(-\tfrac {(c-\mu)^2}{2\sigma^2}) $ for some $c\in [\tilde t ,t]$. It follows that $\left| \Pr_{x\sim\calN(\mu,\sigma^2)}[x<\tilde t] - p\right| \leq |\tilde t -t|\cdot (2\pi\sigma^2)^{-0.5}\cdot 1 \leq \tau\sigma/ \sqrt{2\pi\sigma^2} \leq 0.399\tau$. We comment that this upper bound is fairly tight around the mean, and as long as $\tilde t \in (\mu-\sigma,\mu+\sigma)$ the mean-value theorem also assures us that $\left| \Pr_{x\sim\calN(\mu,\sigma^2)}[x<\tilde t] - p\right| \geq |\tilde t -t| \cdot \tfrac{e^{-0.5}}{\sqrt{2\pi\sigma^2}} \approx 0.24197\tfrac{|\tilde t - t|}{\sigma}=0.24197\tau$.
\end{enumerate}
\end{proof}

\subsection{Detecting Whether the Variance is Large}
\label{subsec:detect_large_variance}
We now present a way to detect whether the unknown variance is large (larger than $2R$) or not (less than $R$).  We start with the following proposition.
\begin{proposition}
\label{pro:probability_mass_outside_interval}
Let $\calD$ be the underlying distribution of the data, hence $\calD = \calN(\mu,\sigma^2)$ for some unknown $\mu$ and $\sigma^2$. Denote $p = \Pr_{X\sim\cal D}[-2R \leq X\leq 2R]$. Then,
\begin{enumerate}
\item If $\sigma \leq R$ then $p > 0.83$.
\item If $\sigma \geq 2R$ then $p < 0.69$.
\end{enumerate}
\end{proposition}
\begin{proof}
Both articles require that we solve an optimization problem. In the first case, our goal is to find $\mu\in [-R,R]$ and $\sigma\leq R$ as to minimize $p$; and in the latter case, our goal is to find $\mu \in [-R,R]$ and $\sigma \geq 2R$ as to maximize $p$. It is evident that the larger $\sigma$ is, the more probability mass falls outside the $[-2R,2R]$ interval, and so in the former case we set $\sigma = R$ and in the latter case we set $\sigma = 2R$. Therefore, since Gaussians are scale invariant, these problems are equivalent to finding $\mu \in [-1,1]$ such that the probability mass on the interval $[-2,2]$ of a Gaussian $\calN(\mu, 1)$ is minimized (for the former case) or that the probability mass on the interval $[-1 ,1]$ is maximized (in the latter case). It is simple to see that maximization is obtained by setting $\mu=0$ and minimization is obtained by setting $\mu=1$ (or $\mu=-1$). Applying now to known quantiles of the Gaussian distribution, we know that $\int\limits_{-1}^{1} (2\pi)^{-0.5}\exp(-\tfrac 1 2 x^2)dx \approx 0.682689 < 0.69$, and that $\int\limits_{-2}^2 (2\pi)^{-0.5}\exp(-\tfrac 1 2 (x-1)^2)dx \approx \tfrac 1 2\cdot  0.682689+  \tfrac 1 2 \cdot 0.9973 > 0.83$.
\end{proof}

Proposition~\ref{pro:probability_mass_outside_interval} gives us a simple recipe for detecting whether the variance is large. We simply use randomized response to estimate the fraction of the population that falls inside the interval $[-2R,2R]$, up to an error of $\leq 7\%$. Based on Claim~\ref{clm:RR_utility_random_input}, we know that w.p. $\geq 1-\beta$ the estimator $\tfrac 1 m \hat{\theta_{\rm RR}}(m, \phi)$ for $\phi(x)= \1\{ x\in [-2R,2R] \}$ is accurate up to $7\%$ if  we apply it to at least $m \geq \frac 2 {0.07^2}\left( \frac{(e^\epsilon+1)^2\log(4/\beta)}{(e^\epsilon-1)^2}\right)$ many samples. If our estimate for this fraction is $\geq 0.76$ then it must be the case that $p>0.69$ and so $\sigma < 2R$, which means we now invoke the algorithm of Section~\ref{sec:unknown_variance_algorithm} using $\sigma_{\max}=2R$; otherwise, it must be the case that $p < 0.83$ and therefore $\sigma > R$, and we continue to deal with the case of really large variance. Seeing as $m = \Theta\left( \frac{(e^\epsilon+1)^2\log(1/\beta)}{(e^\epsilon-1)^2}\right)$ this initial verification increases the overall sample complexity of our algorithm by a $(1+o(1))$-factor.

\subsection{Finding a Confidence Interval for Gaussians of Large Variance}
\label{subsec:algorithm_really_large_unknown_variance}

We now deal with the case where $\sigma > R$. Our plan is fairly simple: we first estimate the probability mass of points $\leq -R$ and then estimate the probability mass of the points $\leq R$, and use the two quantiles to interpolate $\mu$ and a proper confidence interval. In more detail, suppose we know the exact values of $p_- \stackrel{\rm def}= \Pr_{X\sim \calN(\mu,\sigma)}[X \leq -R]$ and $p_+\stackrel{\rm def}= \Pr_{X\sim \calN(\mu,\sigma)}[X \leq R]$, we could use them to interpolate $\mu$ and $\sigma$ in the following way. Denote $t_+$ and $t_-$ as the $p_-$ -quantile  and $p_+$-quantile, respectively, i.e. $\Pr_{Y\sim\calN(0,1)}[Y \leq t_-] = p_-$ and $\Pr_{Y\sim\calN(0,1)}[Y \leq t_+] = p_+$.  Then we have the following
\begin{align}
p_- &\stackrel{}= \Pr_{X\sim \calN(\mu,\sigma)}[X \leq -R] 
\cr & =\Pr_{X\sim \calN(\mu,\sigma)}[\tfrac{X-\mu}\sigma \leq \tfrac{-R-\mu}\sigma] 
\cr &= \Pr_{Y\sim\calN(0,1)}[Y \leq t_-] \cr
\textrm{Similarly, }p_+ &\stackrel{}= \Pr_{X\sim \calN(\mu,\sigma)}[X \leq R] 
\cr &=\Pr_{X\sim \calN(\mu,\sigma)}[\tfrac{X-\mu}\sigma \leq \tfrac{R-\mu}\sigma] 
\cr &= \Pr_{Y\sim\calN(0,1)}[Y \leq t_+] 
\label{eq:quantiles_notations}
\end{align}	
so we deduce that $\tfrac{-R-\mu}\sigma = t_-$ and $\tfrac{R-\mu}\sigma = t_+$, or alternatively, that $\sigma = \tfrac{2R}{t_+-t_-}$ and $\mu= R - \sigma t_+ = R \left(1 - \frac{2t_+} {t_+-t_-}\right) = R\left(\frac  {(-t_-)-t_+}{(-t_-)+t_+} \right)$. Note that since $-R\leq \mu \leq R$ then $t_- \leq 0$ whereas $t_+\geq 0$, hence $\sigma$ has to be positive and we get that $\left| \frac  {(-t_-)-t_+}{(-t_-)+t_+}\right| \leq 1$.

Of course, the fact we have a finite-size sample and use locally private estimators implies we can only estimate $p_-$ and $p_+$ up to a certain error. Thus, the above equalities are replaced with our estimations: $\tilde p_- = p_- + \alpha_-$ and $\tilde p_+ = p_++\alpha_+$, where  both $\alpha_-$ and $\alpha_+$ are bounded in magnitude by $O(\sqrt{\tfrac{\log(1/\beta)}{\epsilon^2 n}})$ (we assume $\epsilon<1$ for this discussion). Denoting $\tilde t_-$ and $\tilde t_+$ as the quantiles\footnote{We understand that as real numbers, we can only approximate $\tilde t_-$ and $\tilde t_+$ rather than have their exact value. However, it is possible to apply standard techniques to approximate those to an error significantly smaller than $1/\sqrt{n}$, and so, for the sake of clarity, we ignore such approximation errors in our discussion.} 
such that 
\begin{align*}
&\Pr_{Y\sim\calN(0,1)}[Y<\tilde t_-] = p_-+\alpha_-
\cr  \textrm{and } &\Pr_{Y\sim\calN(0,1)}[Y<\tilde t_+] = p_++\alpha_+\end{align*}
Claim~\ref{clm:perturbations_of_quantiles} (1) assures us that both $\Delta_- \stackrel{\rm def}=  t_--\tilde t_-$ and $\Delta_+ \stackrel{\rm def}= t_+-\tilde t_+$ have magnitude which is upper bounded by $18.53|\alpha_-|$ and $18.53|\alpha_+|$ respectively, and so both are also in $O(\sqrt{\tfrac{\log(1/\beta)}{\epsilon^2 n}})$. Denote $B$ as an upper bound on the magnitude of $\Delta_-$ and $\Delta_+$. If it were to hold that $(-t_-)+t_+ > 5B$ (we chose the arbitrary constant $5$) then we apply the following inequality:\\$\forall a,b$ s.t. $|a|<|b|$ and $\forall x,y \textrm{ s.t. }\max\{|x|,|y|\}~\leq B~\leq~\tfrac b {5}$, we have 
\begin{align} 
\label{eq:fraction_perturbation}
\left|\frac {a+x}{b+y}-\frac a b \right| &= \left|  \frac{bx-ay}{b(b+y)} \right| \leq \left| \frac x {b+y}\right| + \left| \frac y {b+y}\right| 
\cr &\leq  2\cdot \frac {B}{\frac 4 5|b|}  \leq \frac {3B}{|b|}
\end{align}
and we get that in our case, the following bound holds
\begin{align*} |\tilde \mu - \mu| &\leq R\left| \frac  {(-t_-)-t_+ - \Delta_--\Delta_+}{(-t_-)+t_+-\Delta_- + \Delta_+}-\frac  {(-t_-)-t_+}{(-t_-)+t_+} \right| 
\cr &\leq 3R \frac {B}{(-t_-)+t_+} \end{align*} where $B$ is on the order of $O(\sqrt{\tfrac{\log(1/\beta)}{\epsilon^2 n}})$ and we know that we have that $(-t_-)+t_+ \geq  5B$.

To guarantee that indeed $(-t_-)+t_+ \geq  5B$ we rely on Claim~\ref{clm:perturbations_of_quantiles}, whose corollary implies that the difference $t_+-t_-$ is linearly related to $\Pr_{Y\sim\calN(0,1)}[ Y \in (t_-,t_+)  ]$. In fact, from our definitions we have $\Pr_{Y\sim\calN(0,1)}[ Y \in (t_-,t_+)  ] = \Pr_{X\sim \calN(\mu,\sigma)} [ -R < X  < R] = p_+ - p_-$, and so we have that $t_+-t_- \geq \tfrac 1 {0.399}(p_+-p_-) \geq 2.506(p_+-p_-)$, implying it suffices to verify that $p_+-p_-\geq 2B$. All that is left is to argue that $p_+ - p_- < 2B$ only when $\sigma = \Omega(\sqrt n \cdot \epsilon R)$ in which case the variance is so large that the original interval $[-R,R]$ is a suitable good confidence interval to output.

This discussion gives rise to our very large variance case algorithm (we assume that we have detected already that $\sigma > R$.

\begin{algorithm}
\caption{Very Large Variance Estimator}
\label{alg:large_unknown_variance}
\begin{algorithmic}[1]
\REQUIRE Data $\{ x_1, \cdots, x_n\}$; $\epsilon$, $R$, $\beta>0$.
\STATE Denote $\phi^-(x) = \1\{ x\leq -R \}$, and  $\phi^{+}(x) = \1\{x\leq R\}$.
\STATE Run randomized response on the first $\nicefrac n 2$ datapoints to estimate $\Pr_{X\sim \calN(\mu,\sigma^2)}[ X\leq -R ]$:\\ Set $\tilde p_-\gets \tfrac 2 n \hat{\theta_{\rm RR}}(\tfrac n 2, \phi^{-})$.
\STATE Run randomized response on the latter $\nicefrac n 2$ datapoints to estimate $\Pr_{X\sim \calN(\mu,\sigma^2)}[ X\leq R ]$:\\ Set $\tilde p_+\gets \tfrac 2 n \hat{\theta_{\rm RR}}(\tfrac n 2, \phi^{+})$.
\STATE Denote $B = \sqrt{ \frac {1} n \left(\frac {e^\epsilon+1}{e^\epsilon-1}\right)^2 \log(\tfrac 8 \beta) }$
\IF {$(\tilde p_+ - \tilde p_- < 800B)$}
\RETURN $\tilde \mu=0$ and interval $I=[-R,R]$
\ELSE
\STATE Compute the thresholds $\tilde t_-$ and $\tilde t_+$ s.t. $\Pr_{Y\sim \calN(0,1)}[ Y\leq \tilde t_- ] = \tilde p_-$ and $\Pr_{Y\sim \calN(0,1)}[ Y\leq \tilde t_+ ] = \tilde p_+$
\STATE Set $\tilde \mu \leftarrow R\cdot \frac  {-\tilde t_+-\tilde t_-}{\tilde t_+ - \tilde t_-}$.
\STATE Set $\tau  \leftarrow 9R\cdot \frac{B}{\tilde t_+-\tilde t_-}$.
\RETURN The interval $I = [\tilde \mu - \tau, \tilde{\mu}+\tau]$
\ENDIF
\end{algorithmic}
\end{algorithm}

\begin{theorem}
\label{thm:correnctness_large_variance}
Fix $0<\beta<\tfrac 1 e$. If $\sigma>R$ then w.p. $\geq 1-\beta$ Algorithm~\ref{alg:large_unknown_variance} returns a confidence interval $I$ such that (a) $\mu\in I$ and (b) $|I| \leq 20,000 \sigma \cdot \left(\frac {e^\epsilon+1}{e^\epsilon-1}\right)\sqrt{\frac{\log(8/\beta)}{n}}$.
\end{theorem}
\begin{proof}
First, denoting $p_- = \Pr_{X\sim \calN(\mu,\sigma^2)}[ X\leq -R ]$ and $p_+= \Pr_{X\sim \calN(\mu,\sigma^2)}[ X\leq R ]$,  we apply Claim~\ref{clm:RR_utility_random_input} to infer that w.p. $\geq 1-\beta$ we have that both $|\tilde p_- - p_-| \leq \frac{2}{\sqrt n}\cdot \frac {e^\epsilon-1}{e^\epsilon-1} \sqrt{\log(\tfrac 8 \beta)} =2B$ and $|\tilde p_+ - p_+| \leq \frac{2}{\sqrt n}\cdot \frac {e^\epsilon-1}{e^\epsilon-1} \sqrt{\log(\tfrac 8 \beta)}=2B$. Thus, if it indeed holds that $\tilde p_+ - \tilde p_- <800B$ then it must be that $p_+-p_- < 804B$. In other words, we know that $\Pr_{X\sim\calN(\mu,\sigma^2)}[X \in [-R,R]   ] \leq 804B$. Claim~\ref{clm:perturbations_of_quantiles} (1) assures us that in this case it must be that the distance $R-(-R) \leq 20\sigma \cdot 804B$, i.e. $2R\leq 20,000B\cdot \sigma$. As we return in this case the interval $[-R,R]$ which is guaranteed to contain $\mu$, the required holds.

We therefore turn to analyze the interval returned by Algorithm~\ref{alg:large_unknown_variance} in the case of $\tilde p_+-\tilde p_- \geq 800B$. In this case, seeing as we already conditioned on $\max\{ |\tilde p_- - p_-|,|\tilde p_+ - p_+|  \} \leq 2B$, by Claim~\ref{clm:perturbations_of_quantiles}(1) it must hold that 
\[ \max\{ |\tilde t_+ - t_+|, |\tilde t_--t_-|   \}\leq 40B  \]
yet on the other hand, $p_+-p_- \geq (800-2-2)B$ and therefore, by Claim~\ref{clm:perturbations_of_quantiles}(2) it must be that $t_+$ and $t_-$ are sufficiently far apart to allow a probability mass $\geq 796B$ to fall in the interval $[-R,R]$; I.e. $t_+ -t_- > 0.39\times 796B > 300B > 5\cdot 40B$. Thus, all the conditions of \eqref{eq:fraction_perturbation} hold and we have that $|\tilde \mu - \mu| \leq 3B \cdot \frac R {t_+-t_-} = \frac {3B} 2 \cdot \frac {2R}{t_+-t_-}$. Recall that by definition: $t_-=\frac {-R-\mu}{\sigma}$ and $t_+=\frac {R-\mu}{\sigma}$ (see \eqref{eq:quantiles_notations}) hence $\frac {2R}{t_+-t_-}=\sigma$, and so it holds that $|\tilde \mu - \mu| \leq 1.5B\sigma$. Lastly, re-applying the same reasoning of \eqref{eq:fraction_perturbation}, we can infer that
\[1.5B\sigma = \frac {3B} 2\cdot \frac {2R}{t_+- t_-} = 3R \frac{B}{t_+- t_-} \leq 3R\frac {3B}{\tilde t_+- \tilde t_-}=\tau\] 
proving that indeed w.p. $\geq 1-\beta$ we have that $\mu \in (\tilde \mu - \tau, \tilde{\mu}+\tau)$ and the interval we return, whose length is $2\tau$ satisfies the desired bound.
\end{proof}

\section{Lower Bounds}
\label{sec:lower_bounds}

We begin our discussion on the bounds for the utility of any $\epsilon$-locally private mechanism which is a one-shot mechanism, by presenting the following lemma. This lemma is a combination of two separate results. First, Karwa and Vadhan's coupling argument that suggest that the ``effective group privacy'' between two $n$-size samples from either a distribution $\calP$ or a distribution $\calQ$ is roughly $n\cdot d_{\rm TV}(\calP,\calQ)$. The second result is a lemma, which originally appeared in~\citet{BNO08} and then also appeared in a more formal way in~\citet{BNS18}, that states that group privacy of altering $k$ datums scales proportional to $O(\epsilon\sqrt k)$ rather $O(\epsilon k)$ as in the centralized model. We combine the two into a single lemma, dealing with $\epsilon$-LDP mechanisms over an input that is drawn iid from some distribution. This lemma is the main building block in all of our lower-bounds.

\begin{lemma}
\label{lem:local_DP_random_inputs}
Let \calM ~be a one-shot local $\epsilon$-differentially private mechanism. Let $\calP$ and $\calQ$ be two distributions, with $\Delta\stackrel{\rm def} = d_{\rm TV}(\cal P,\cal Q)$. Fix any $0<\delta<e^{-1}$ and set $\epsilon^*=8\epsilon\Delta \sqrt{n} \left( \sqrt{ 2 \log(\nicefrac 2 \delta)} + 16\epsilon\Delta \sqrt n \right)$. Then, for any set of possible outputs $S$ we have that
	\[ \Pr_{ \bX \stackrel{\rm i.i.d}\sim \calP  }[ \calM( \bX \in S  ] \leq e^{\epsilon^*}\Pr_{\bX \stackrel{\rm i.i.d}\sim \calQ  }[ \calM( \bX) \in S  ]  + \delta \] where the probability is taken over both the $n$ i.i.d samples and over the coin-tosses of $\calM$.
\end{lemma}

The proof of Lemma~\ref{lem:local_DP_random_inputs} is based on the following fact.

\begin{fact}
\label{fact:DP_implies_bounded_KL}
Let $\calP$ and $\calQ$ be two distributions over the same domain $\calX$, such that there exists a bound $B>0$ so that for any $x\in \calX$ we have that $\left| \log\left(\frac{\PDF_\calP[x]}{\PDF_\calQ[x]} \right) \right|\leq B$. Then \[d_{\rm KL}(\calP;\calQ) \leq d_{\rm KL}(\calP;\calQ)+d_{\rm KL}(\calQ;\calP) \leq B(e^B-1)\]
\end{fact}
The proof of Fact~\ref{fact:DP_implies_bounded_KL} appears in full detail in~\citet{DR14}.

As an immediate corollary of Fact~\ref{fact:DP_implies_bounded_KL} and Azuma's Inequality, we have the~\citet{BNS18} group-privacy in the (one-shot) local model. Let $S$ be the set of indices on which the inputs $D$ and $D'$ differ. We denote by $x_{1:t}$ the first $t$ entries of $D$ and by $x'_{1:t}$ the first $t$ entries of $D'$, and by $\calM_{1:t}(x_{1:t})$ (resp. $\calM_{1:t}(x'_{1:t})$) the outcome of the mechanism over the first $t$ entries from $D$ (resp. from $D'$). For each $i\in S$ we denote $X_i =  \log(\frac{\PDF[\calM_i(x_i)]}{\PDF[\calM_i(x'_i)]})$ and $Y_i = \left(X_i | X_{1:i-1}\right)$. Recall, the $i$th interaction with the $i$th user might dependent on the previous interactions with previous users and so we must condition on the previous $i-1$ results.\footnote{We thank Steven Wu for bringing this subtlety to our attention.} It is simple enough to see that $Y_i- \E[X_i|~X_{i-1}]$ is a martingale, that due to $\epsilon$-differential privacy we have that $|Y_i|\leq \epsilon$. Fact~\ref{fact:DP_implies_bounded_KL} implies that $\E[Y_i] \leq \epsilon(e^\epsilon-1)$, and furthermore Azuma's inequality gives that the sum of the privacy losses exceeds its mean by more than $t=\epsilon \sqrt{2|S|\log(2/\delta)}$ w.p. of at most $\delta$. Thus, w.p. $\geq 1-\delta$ our privacy loss is bounded by $|S|\epsilon(e^\epsilon-1) + \epsilon \sqrt{2|S|\log(2/\delta)}$.

We aim to give a similar bound, but under the assumption that the $n$ entries in the dataset are changing by resampling them from a distribution. 

\begin{proof}[Proof of Lemma~\ref{lem:local_DP_random_inputs}]
We mimic the proof of \citet{KV18}. Let $f(x) = \max\{ \PDF_\calP[x]-\PDF_\calQ[x], 0 \}$, $g(x) = \max\{ \PDF_\calQ[x]-\PDF_\calP[x], 0 \}$, and $h(x) = \min\{ \PDF_\calP[x], \PDF_\calQ[x] \}$. Note that $f$ and $g$ integrate to $\Delta$ and that $h$ integrates to $1-\Delta$. Let $F,G,H$ be the normalizations of $f,g$ and $h$ respectively such that all three are distributions. We now generate a coupling of the two distributions; namely, we describe a process that generates pairs of inputs $(x_i, x'_i)$. For each $i$ from $1$ to $n$ we
\begin{enumerate}
\item Pick a bit $b_i$ such that $\Pr[b_i=1]=\Delta$ and ${\Pr[b_i=0]=1-\Delta}$.
\item If $b_i = 0$ then we sample $x_i \sim H$ and set $x'_i = x_i$.
\item If $b_i = 1$ then we sample $x_i \sim F$ and $x'_i \sim G$ independently. 
\end{enumerate}
It is simple enough to verify that $x_i \sim \calP$ and $x'_i \sim \calQ$. 

Let $X_i$ be that privacy loss under $\calM_i$ of the $i$-th datum under this coupling. Namely, $X_i = \log( \frac{\PDF[\calM_i(x_i)]}{\PDF[\calM_i(x'_i)]} )$. We claim that $|X_i| \leq 8\epsilon\Delta$.
\begin{align*}
{X_i} &= \log\left( \frac{\PDF[\calM_i(x_i)]}{\PDF[\calM_i(x'_i)]} \right)
\cr &=\log\left( \frac{(1-\Delta)\cdot\PDF_H[\calM_i(x_i)] + \Delta\cdot\PDF_F[\calM_i(x_i)]}{(1-\Delta)\cdot\PDF_H[\calM_i(x'_i)] + \Delta\cdot\PDF_G[\calM_i(x'_i)]}\right)
\cr & \leq  \log\left(\frac{(1-\Delta)\cdot\PDF_H[\calM_i(x_i)] + e^\epsilon\Delta\cdot\PDF_H[\calM_i(x_i)]}{(1-\Delta)\cdot\PDF_H[\calM_i(x'_i)] + e^{-\epsilon}\Delta\cdot\PDF_H[\calM_i(x'_i)]} \right)
\cr & = \log\left(\frac{ 1+\Delta(e^\epsilon-1) }{1 - \Delta(1-e^{-\epsilon})} \right)
\cr &= \log\left(1 + \frac{  \Delta(e^\epsilon-e^{-\epsilon}) }{1 - \Delta(1-e^{-\epsilon})  }\right) \leq \log\left(1 + {  \Delta(e^\epsilon-e^{-\epsilon}) }\right) 
\cr &\leq \Delta(e^\epsilon-e^{-\epsilon})
\intertext{Similarly,}
X_i &\geq \log\left( \frac{1-\Delta(1-e^{-\epsilon})}{1+\Delta(e^\epsilon-1)}\right) 
\cr &=  \log\left( 1 - \frac{\Delta(e^\epsilon-1 +1-e^{-\epsilon})}{1+\Delta(e^\epsilon-1)} \right)  
\cr &\geq  \log\left(1-\Delta(e^\epsilon-e^{-\epsilon}) \right) \geq -2\Delta(e^\epsilon-e^{-\epsilon}) 
\end{align*}
where the last inequality holds for sufficiently small values of $\epsilon$.\\
For $\epsilon <1$ we have that $e^\epsilon-e^{-\epsilon} < 4\epsilon$ resulting in the desired bound: $|X_i|\leq 8\epsilon \Delta$.

Plugging this into the result of ``group privacy'' discussed above, we have that replacing all $n$ datums from sampled given by $\calP$ to samples of $\calQ$ we have that w.p. $\geq 1-\delta$ the privacy loss is at most
\begin{align*}  
&n(8\epsilon\Delta)(e^{8\epsilon\Delta}-1) + 8\epsilon\Delta \sqrt{2n\log(2/\delta)} 
\cr &~~~~~~~\leq 8\epsilon\Delta \sqrt{n} \left( \sqrt{2 \log(2/\delta)} + 16\epsilon\Delta \sqrt n \right) \end{align*}
assuming $\epsilon<1$ hence $e^\epsilon<1+2\epsilon$.
\end{proof}

\subsection{Lower Bounds for One-Shot $\epsilon$-Locally Private Mechanisms}
\label{subsec:LBs}

Leveraging on our main lemma, we can now prove lower bounds on the interval length and sample complexity of any one-shot $\epsilon$-LDP algorithm that outputs a meaningful confidence interval. We focus on the case of a known variance, and our lower-bound shows the optimality of Algorithm~$\KnownBF$ up to a $O(\sqrt{\log(\nicefrac n \beta)})$-factor.

\DeclareRobustCommand{\thmLDPConfidenceIntervalLB}{We say an algorithm \emph{$(\beta,\tau)$-solves the mean-estimation problem} (under known variance $\sigma^2$ and bound $R$) if its input is a sample of $n$ points and its output is an interval $I$ such that, if all $n$ datapoints are iid draws from $\calN(\mu,\sigma^2)$ for some $\mu\in [-R,R]$ then w.p. $\geq 1-\beta$ it holds that $\mu\in I$ and furthermore, $\E[|I|] \leq \tau$. (The probability is taken over \emph{both} the sample draws and the coin-tosses of the algorithm.)
	
	Fix any $\beta<\nicefrac 1 3$. Then any one-shot $\epsilon$-locally differentially private algorithm $\calM$ that $(\beta,\tau)$-solves that mean estimation problem must have that $\tau = \Omega\left( \tfrac{\sigma\sqrt{\log(1/\beta)}}{\epsilon\sqrt{n}}\right)$ and also that $n=\Omega\left( \tfrac 1 {\epsilon^2} \log( \tfrac R{\beta\cdot\tau} )\right)$.}
\begin{theorem}
\label{thm:main_LB}
\thmLDPConfidenceIntervalLB
\end{theorem}
\begin{proof}
To prove Theorem~\ref{thm:main_LB} we consider the following problem, defined by a parameter $\tau$. We define the following collection of points on the interval $[-R,R]$: $P = \{ p_0 = -R, p_1 = -R + 2\tau, ... p_i = -R + 2i\tau, ..., p_m= R \}$. Clearly, this is a collection of $\lfloor \tfrac {2R}{2\tau}\rfloor+1 =\lfloor \tfrac R \tau\rfloor+1\stackrel{\rm def}=m+1$ possible outputs.  (For simplicity we assume $R$ is divisible by $\tau$ otherwise, we set the later point to be the nearest integer multiplication of $2\tau$). We say an algorithm is $\beta$-useful for if on any sample of $n$ iid draws from a Gaussian of variance $\sigma^2$ and mean $p_i$, the algorithm returns the correct index $i$ w.p. $\geq 1-\beta$ (over the draws \emph{and} the coin tosses of the algorithm.
Clearly, if there exists a one-shot $\epsilon$-DP in the local model algorithm that $(\beta,\tau)$-solves the mean estimation problem, then using it as a black-box we can design a $\beta$-useful algorithm for the above problem. We thus proceed to argue that no one-shot $\epsilon$-DP in the local model is $\beta$-useful for the above problem unless $n$ is sufficiently large.

The argument we invoke is the standard packing-argument. Let $\calP_i$ be the distribution of Gaussian of mean $p_i$ and variance $\sigma^2$. Let $\calM$ be any $\beta$-useful one-shot $\epsilon$-DP in the local model for the above problem. It follows that, for any choice of $\delta_1, \delta_2,...,\delta_m$ and the respectively defined $\epsilon_1^*,...,\epsilon_n^*$ given by Lemma~\ref{lem:local_DP_random_inputs}, we get
\begin{align}
\beta &\geq \Pr_{\bX \stackrel{iid}\sim \calP_0, \calM}[\calM(\bX) \neq 0] \geq \sum_{i>0} \Pr_{\bX \stackrel{iid}\sim \calP_0, \calM}[\calM(\bX) =i]
\cr &\stackrel{\rm Lem~\ref{lem:local_DP_random_inputs}}\geq \sum_i \exp\left(-\epsilon_i^*\right)\Pr_{\bX \stackrel{iid}\sim \calP_i, \calM}[\calM(\bX) =i] - \delta_i
\cr &\stackrel{}\geq (1-\beta)\left(\sum_i \exp\left(-\epsilon_i^*\right)\right) - \delta_i \label{eq:general_lower_bound}
\end{align}
From Equation~\eqref{eq:general_lower_bound} we derive multiple conclusions, using also the fact that
\[  \forall i>0,  ~~~\Delta_i = d_{\rm TV}(\calP_i, \calP_0)\leq \min\{1, \frac{2\tau\cdot i} {\sigma} \} \]
and denote $m_0 = \max \{ i:~ \tfrac{2\tau i} \sigma \leq 1 \}$ (thus, $m_0 \leq \tfrac \sigma{2\tau}$).

First of all, $\beta$ is lower bounded by at least the first of the summands in Equation~\eqref{eq:general_lower_bound}. Setting $\delta_1 = \tfrac \beta 4$ we get
\begin{align*}
&\beta\geq (1-\beta)\exp(-16\epsilon\tfrac\tau \sigma \sqrt{n} \left( \sqrt{2 \log(8/\beta)} + 32\epsilon\tfrac \tau\sigma \sqrt n \right)) - \tfrac \beta 4
\intertext{so,}
 &\tfrac {5\beta}{4(1-\beta)} \geq  \exp(-16\epsilon\tfrac\tau \sigma \sqrt{n} \left( \sqrt{2 \log(8/\beta)} + 32\epsilon\tfrac \tau\sigma \sqrt n \right))
\end{align*}
Therefore, we conclude that it must hold that $\tau \geq \tfrac 1 {20\cdot 16} \cdot \tfrac \sigma{\epsilon\sqrt n} \sqrt{\log(8/\beta)}$ for otherwise we get
\begin{align*} \tfrac 3 2 \cdot \tfrac{5\beta}4 &\geq \tfrac{5\beta}{4(1-\beta)} 
\cr &\geq \exp\left(-\tfrac 1 {20} \sqrt{\log(\tfrac 8 \beta)} ( \sqrt{2 \log(\tfrac 8\beta)} + \tfrac 1 {20} \sqrt{\log(\tfrac 8 \beta} )\right)  
\cr &\geq \exp\left(-\tfrac 1 {10} \log(\tfrac 8 \beta)\right) = \sqrt[10]{\tfrac \beta 8} \end{align*} contradicting the fact that $\beta \leq \tfrac 1 3$. As a result we get that $n \geq \tfrac{\sigma^2}{\tau^2} \tfrac {\log(8/\beta)} {160^2\epsilon^2}$. Hence, if $\tau\leq \sigma$ we get $\log(8/\beta) \leq (160\epsilon)^2 n$; and if $\tau > \sigma$ we can repeat the above derivation only now using $1$ as the upper bound on $\Delta_1$ and still have that $n\geq \Omega\left(\frac{\log(1/\beta)}{\epsilon^2}\right)$.

Secondly, we can set all $\delta_i = \tfrac \beta {4m}$ and use the fact that $\Delta_i\leq 1$ for all $i$ to get a lower bound of the form
\begin{align*}
 \beta &\geq (1-\beta)m e^{\left(  -8 \epsilon \sqrt n ( \sqrt{2 \log(\tfrac {8m} \beta)} + 16 \epsilon\sqrt n )  \right)} - m\cdot \tfrac \beta{4m} 
 \end{align*}
 It follows that if $n\leq \frac{\log(8m/\beta)}{200\cdot 16^2\epsilon^2}$ we get that
 \begin{align*}
\tfrac {15 }8 \beta &\geq  \tfrac {5\beta}{4(1-\beta)} 
\cr &\geq m \exp\left(  -\tfrac 1 {10} \sqrt{\log(\tfrac{8m}{\beta})}\cdot (\sqrt{\tfrac 1 2} +\tfrac 1 {10}) \sqrt{\log(\tfrac{8m}{\beta})} \right) 
\cr &= m\exp(-\tfrac 1 {10} {\log(\tfrac{8m}{\beta})})
\cr &= m\cdot (\tfrac \beta{8m})^{1/10} = m^{0.9}(\tfrac \beta{8})^{0.1}
\end{align*}
which, using the fact that $\beta\leq 1/3$, contradicts the fact that $m>1$. (If $R\leq 2\tau$ then the problem has only a single solution.) 
\end{proof}

It is worth-while to discuss the implications of Theorem~\ref{thm:main_LB}. Aside from showing the near optimality of our technique, it also shows that our dependency on $R$ is of the essence. This is in sharp contrast to the \emph{centralized}-model, when the results of~\citet{KV18} show that there exists a $(\epsilon,\delta)$-differentially private algorithm whose sample complexity is independent of $R$. Our lower bounds, as shown by~\citet{BNS18}, from the $\epsilon$-LDP setting carry over to the $(\epsilon,\delta)$-LDP, so that the same dependency on $\log(R)$ is required. This illustrates a sharp contrast between the centralized and the local model.

In addition, we prove a similar bound on the optimality of the $\BinRR$-Algorithm.

\DeclareRobustCommand{\thmLDPQuantileLB}{Let $\calM$ be a $\epsilon$-LDP mechanism which is $(\adist,\aquant,\beta)$-useful for the $p$-quantile problem over $\calP$, given that the true $p$-quantile lies in the interval $[-R,R]$.
	Then, for any $\beta<\tfrac 1 {6}$ it must hold that $n\geq \Omega(\tfrac 1 {\aquant^2\epsilon^2}\cdot\log(\tfrac{R}{\adist\beta}))$.}
\begin{theorem}
\label{thm:LB_quantiles}
\thmLDPQuantileLB
\end{theorem}
\begin{proof}
Similar to the proof of Theorem~\ref{thm:main_LB}, we define a collection of distributions $\calP_1,...,\calP_m$ s.t. a good answer for data drawn from $\calP_i$ is necessarily a bad answer for any $\calP_j$ for $j\neq i$. Also similar to the proof of Theorem~\ref{thm:main_LB}, our construction is also based on the collection $C = \{ t_0 = -R, t_1 = -R + 2\adist, ... t_i = -R + 2i\adist, ..., t_m= R \}$, and note how $|C|=m+1 = \lceil \tfrac{2R}{2\adist} \rceil$. Given a quantile $p\in (\aquant,1-\aquant)$,\footnote{if $p<\aquant$ then $\calM$ may return $\tilde t=-R$ without looking at any sample, and similarly return $\tilde t = R$ if $p>1-\aquant$.} we denote the suitable $\calP_i$ defined as a discrete distribution over $3$ points: $\{-R,t_i,R\}$. (For $i=0$ or $i=m$, with $t_i=-R$ or $t_i=R$ resp., we just sum the probability of falling at the extreme or at $t_i$.)
\begin{align*} 
&\Pr_{X\sim\calP_i}[X=-R]=q -\aquant 
\cr & \Pr_{X\sim\calP_i}[X=t_i]=2\aquant
\cr & \Pr_{X\sim\calP_i}[X=R]=1-q-\aquant  \end{align*}
Clearly, for each $i$, the $q$-quantile of $\calP_i$ is $t_i$. It follows that when all $n$ datums are drawn from $\calP_i$ then w.p. $\geq 1-\beta$ it must be that $\calM$ returns an answer which is in the interval $t_i \pm\adist$. As the distance between any two distinct $t_i$ and $t_j$ is $2\adist$ then it is evident that an answer in $t_i\pm\adist$ is bad for the case where the data is drawn from $\calP_j$. Lastly, it is also simple to see that for any $i\neq j$ we have that $d_{\rm TV}(\calP_i, \calP_j) = 2\aquant$. 

We apply Lemma~\ref{lem:local_DP_random_inputs} fixing for all $i$, $\delta_i=\tfrac{\beta}{2m}$  and $\epsilon_i=\epsilon$, thus $\epsilon^*_i = 8\epsilon\cdot 2\aquant\cdot \sqrt{n}\left(\sqrt{\tfrac 1 2 \log(\tfrac{8m}\beta)} + 16\epsilon \cdot 2\aquant\cdot\sqrt n \right)$. And so,
\begin{align*}
\beta &\geq \Pr_{\bX\stackrel{\rm iid}{\sim} \calP_0;~\calM}[\calM(\bX) \geq -R+\adist] 
\cr& \geq \sum_{i>0}\Pr_{\bX\stackrel{\rm iid}{\sim} \calP_0;~\calM}[\calM(\bX) \in (t_i-\adist, t_i+\adist)]\cr 
&\geq \sum_{i\geq 0} \left(e^{-\epsilon^*}\Pr_{\bX\stackrel{\rm iid}{\sim} \calP_i;~\calM}[\left|\calM(\bX) -t_i\right| \leq \adist] - \tfrac{\beta}{4m}\right) \cr
&\geq (1-\beta)m e^{ -16\epsilon\aquant \sqrt{n}\left(\sqrt{\tfrac 1 2 \log(\tfrac{8m}{\beta})} + 32\epsilon \aquant\sqrt n \right)} - \tfrac\beta 4 \cr
\intertext{As a result, if $n< \tfrac{\log(\tfrac{8R}{\adist\beta})}{2^{12}\aquant^2\epsilon^2}$ then we get}
\tfrac {3\beta}2 &= \tfrac{5\beta}{4\cdot\frac 5 6} \geq \tfrac{5\beta}{4(1-\beta)} 
\cr &\geq \tfrac{R}{\adist} e^{ -16\epsilon\aquant \sqrt{n}\left(\sqrt{\tfrac 1 2 \log(\tfrac{8R}{\adist\beta})} + 32\epsilon \aquant\sqrt n \right)}
\cr &> \tfrac{R}{\adist} e^{ -\tfrac 1 4 \sqrt{\log(8R/\adist\beta)}\left(\sqrt{\tfrac 1 2 \log(8R/\adist\beta)}+\tfrac 1 2\sqrt{\log(8R/\adist\beta)}\right)} 
\cr &= \tfrac{R}{\adist} e^{-\log(8R/\adist\beta)\cdot \left(\tfrac 1 4 (\sqrt{\tfrac 1 2}+\tfrac 1 2)\right)}
\cr & >\tfrac{R}{\adist} \exp(-\tfrac 1 3\log(8R/\adist\beta) )
\cr &> \tfrac{R}{\adist} \exp(-\log(R/\adist)-\tfrac 1 3\log(8/\beta) )
\cr & = \tfrac R{\adist} \cdot \tfrac \adist R \cdot \left(\tfrac\beta 8\right)^{1/3} = \tfrac{\beta^{1/3}} 2 
\end{align*}
Thus, $\beta^{2/3} > \tfrac 1 {3}$ implying $\beta>0.19$ contradicting the fact that $\beta<\tfrac 1 6$.
\end{proof}

It is important to note that our lower bound shows how \emph{all} three parameters are necessary for devising a suitalbe $\epsilon$-LDP algorithm for the problem. For example, we must have both stopping conditions ($\aquant$ and $\adist$). If we didn't specify $\adist$ as well, then we could devise a collection of infinitely many distributions~--- for any point $z\in [-R,R]$ we would construct a similar $\calP_{z}$ similar to $\calP_i$~--- resulting in infinite sample complexity. Then for any $m$ we could create a $m$-size collection of distributions by repeating the same collection with $R$ set to be any number $>m/\adist$, thus we could get a sample complexity as arbitrary large as we want. Lastly, if $\aquant$ was unspecified, we could derive an arbitrarily large sample complexity even without privacy
as finding the exact quantile of a distribution requires infinitely many samples. 

\newcommand{\remove}[1]{}
\remove{
\hrule\hrule\hrule
Failed attempts at getting a $\tau \geq \sqrt{\log(n)/n}$ bound.

Finally, we apply the above two lower bounds on $n$ to Equation~\eqref{eq:general_lower_bound}, by setting $\delta_i = \tfrac \alpha {4m}$ for all $i$, but applying the tighter of the two bounds on $\Delta_i$:
\begin{align*} \alpha &\geq (1-\alpha)\sum_{i=1}^{m_0} \exp\left(  -16i\tfrac \tau \sigma \epsilon \sqrt n ( \sqrt{\tfrac 1 2 \log(\tfrac {8m} \alpha)} + 32i\tfrac \tau \sigma \epsilon\sqrt n )  \right) 
\cr &~~ + (1-\alpha)(m-m_0)\exp\left(  -8 \epsilon \sqrt n ( \sqrt{\tfrac 1 2 \log(\tfrac {8m} \alpha)} + 16 \epsilon\sqrt n )\right) - m\cdot \tfrac \alpha {4m} \cr 
\tfrac {5\alpha}{4(1-\alpha)} &\geq \sum_{i=1}^{m_0} \exp\left(  -16i\tfrac \tau \sigma \epsilon \sqrt n ( \sqrt{\tfrac 1 2 \log(\tfrac {8} \alpha)} +\sqrt{\tfrac 1 2 \log(m)} + 32i\tfrac \tau \sigma \epsilon\sqrt n )  \right) 
\cr &~~ + (m-m_0)\exp\left(  -8 \epsilon \sqrt n ( \sqrt{\tfrac 1 2 \log(\tfrac {8} \alpha)} +\sqrt{\tfrac 1 2 \log(m)} + 16 \epsilon\sqrt n )\right) 
\cr &\geq \sum_{i=1}^{m_0} \exp\left(  -16i\tfrac \tau \sigma \epsilon \sqrt n ( \tfrac{160}{\sqrt 2} \epsilon\sqrt n +\tfrac{130}{\sqrt 2} \epsilon\sqrt n + 16 \epsilon\sqrt n )  \right) 
\cr &~~ + (m-m_0)\exp\left(  -8 \epsilon \sqrt n ( \tfrac{160}{\sqrt 2} \epsilon\sqrt n + \tfrac{130}{\sqrt 2} \epsilon\sqrt n+16 \epsilon\sqrt n )\right) 
\cr &\geq \left(\sum_{i=1}^{m_0} \exp\left(  -i\cdot 4000\tfrac \tau \sigma \epsilon^2 n\right)\right) + (m-m_0) \exp(-2000\epsilon^2 n)
\end{align*}
Note that the lower bound on $\tau$ implies that $4000\tfrac \tau \sigma \epsilon^2 n \geq K$ for some constant $K>0$. 
\begin{proposition}
Given $K>0$ and a natural $s\leq \tfrac 1K$ we have that $\sum_{i=1}^{s} e^{-iK} \geq \tfrac{s}{2}e^{-K}$.
\end{proposition}
\begin{proof}
\begin{align*}
\sum_{i=1}^s e^{-iK} &= \tfrac {e^{-K}}{1-e^{-K}} - \tfrac {e^{-(s+1)K}}{1-e^{-K}} = \tfrac{  e^{-K} }{1-e^{-K}} \left( 1-e^{-sK} \right) \stackrel{(\ast)}\geq e^{-K} \cdot \tfrac{sK/2}{K} \geq \tfrac {s}{2}e^{-K}
\end{align*}
where $(\ast)$ follows from the fact that for $0< sK\leq 1$ we have $e^{-sK} \leq 1-\tfrac {sK}2$.
\end{proof}
We thus split our analysis into two cases, based on the greater of the two terms: $m_0$ or $\frac 1 {4000\tfrac \tau \sigma \epsilon^2 n}$.

Case 1: $m_0 \geq \frac 1 {4000\tfrac \tau \sigma \epsilon^2 n}$. 
}

\clearpage
\bibliographystyle{plainnat}
\bibliography{refs}

\end{document}